\documentclass[12pt,a4paper]{article}
\usepackage[all]{xy}
\usepackage{amsthm}
\usepackage{times}
\usepackage{epsfig}
\usepackage{graphicx}
\usepackage{amsmath}
\usepackage{enumerate}
\usepackage[psamsfonts]{amssymb}
\usepackage{amsfonts}
\usepackage[psamsfonts]{eucal}
\usepackage{txfonts}
\usepackage{mathrsfs}
\usepackage{bm}
\usepackage{color}
\definecolor{darkgreen}{rgb}{0,.6,0}
\definecolor{darkblue}{rgb}{0,0,0.6}
\usepackage[breaklinks=true,dvipdfm,letterpaper=true,colorlinks=true,bookmarks=true,bookmarksopen=true,linkcolor=darkblue,citecolor=darkgreen]{hyperref}

\usepackage[top=33.5mm, bottom=33.5mm, left=32mm, right=32mm]{geometry}

\allowdisplaybreaks[1]

\theoremstyle{definition}
\newtheorem{defn}{Definition}

\theoremstyle{plain}
\newtheorem{lem}{Lemma}
\newtheorem{cor}{Corollary}
\newtheorem{prop}{Proposition}
\newtheorem{thm}{Theorem}

\newcommand{\diagram}[2][]{\begin{aligned} \xymatrix#1{#2} \end{aligned}}

\newcommand{\au}[3][]{\ar@{.>}#1[#2]^{#3}}
\newcommand{\av}[3][]{\ar@{.>}#1[#2]_{#3}}
\newcommand{\aul}[4][]{\ar@{.>}#1[#2]^(#3){#4}}
\newcommand{\avl}[4][]{\ar@{.>}#1[#2]_(#3){#4}}

\newcommand{\aux}[3][]{\ar@{.)}#1[#2]^{#3}}
\newcommand{\avx}[3][]{\ar@{.)}#1[#2]_{#3}}
\newcommand{\auxl}[4][]{\ar@{.)}#1[#2]^(#3){#4}}
\newcommand{\avxl}[4][]{\ar@{.)}#1[#2]_(#3){#4}}

\newcommand{\aiul}[1]{\ar@{.>}@!+<-1.5em,1.5em>;[0,0]^(.2){#1}}
\newcommand{\aidl}[1]{\ar@{.>}@!+<-1.5em,-1.5em>;[0,0]^(.2){#1}}
\newcommand{\aiur}[1]{\ar@{.>}@!+<1.5em,1.5em>;[0,0]_(.2){#1}}
\newcommand{\aidr}[1]{\ar@{.>}@!+<1.5em,-1.5em>;[0,0]_(.2){#1}}
\newcommand{\ajul}[1]{\ar@{.>}@!+<-1.5em,1.5em>;[0,0]_(.2){#1}}
\newcommand{\ajdl}[1]{\ar@{.>}@!+<-1.5em,-1.5em>;[0,0]_(.2){#1}}
\newcommand{\ajur}[1]{\ar@{.>}@!+<1.5em,1.5em>;[0,0]^(.2){#1}}
\newcommand{\ajdr}[1]{\ar@{.>}@!+<1.5em,-1.5em>;[0,0]^(.2){#1}}

\newcommand{\aiulx}[1]{\ar@{.)}@!+<-1.5em,1.5em>;[0,0]^(.2){#1}}
\newcommand{\aidlx}[1]{\ar@{.)}@!+<-1.5em,-1.5em>;[0,0]^(.2){#1}}
\newcommand{\aiurx}[1]{\ar@{.)}@!+<1.5em,1.5em>;[0,0]_(.2){#1}}
\newcommand{\aidrx}[1]{\ar@{.)}@!+<1.5em,-1.5em>;[0,0]_(.2){#1}}
\newcommand{\ajulx}[1]{\ar@{.)}@!+<-1.5em,1.5em>;[0,0]_(.2){#1}}
\newcommand{\ajdlx}[1]{\ar@{.)}@!+<-1.5em,-1.5em>;[0,0]_(.2){#1}}
\newcommand{\ajurx}[1]{\ar@{.)}@!+<1.5em,1.5em>;[0,0]^(.2){#1}}
\newcommand{\ajdrx}[1]{\ar@{.)}@!+<1.5em,-1.5em>;[0,0]^(.2){#1}}

\newcommand{\arloopr}[1]{\ar@{.>}@(ur,dr)!R+<0em,.25em>;!R+<0em,-.25em>^{#1}}
\newcommand{\arloopl}[1]{\ar@{.>}@(ul,dl)!L+<0em,.25em>;!L+<0em,-.25em>_{#1}}
\newcommand{\arloopd}[1]{\ar@{.>}@(dl,dr)!D+<-.25em,0em>;!D+<.25em,0em>_{#1}}
\newcommand{\arloopu}[1]{\ar@{.>}@(ur,ul)!U+<.25em,0em>;!U+<-.25em,0em>_{#1}}

\newcommand{\arlooprv}[1]{\ar@{.>}@(ur,dr)!R+<0em,.25em>;!R+<0em,-.25em>_{#1}}
\newcommand{\arlooplv}[1]{\ar@{.>}@(ul,dl)!L+<0em,.25em>;!L+<0em,-.25em>^{#1}}
\newcommand{\arloopdv}[1]{\ar@{.>}@(dl,dr)!D+<-.25em,0em>;!D+<.25em,0em>^{#1}}
\newcommand{\arloopuv}[1]{\ar@{.>}@(ur,ul)!U+<.25em,0em>;!U+<-.25em,0em>^{#1}}

\newcommand{\arlooprx}[1]{\ar@{.)}@(ur,dr)!R+<0em,.25em>;!R+<0em,-.25em>^{#1}}
\newcommand{\arlooplx}[1]{\ar@{.)}@(ul,dl)!L+<0em,.25em>;!L+<0em,-.25em>_{#1}}
\newcommand{\arloopdx}[1]{\ar@{.)}@(dl,dr)!D+<-.25em,0em>;!D+<.25em,0em>_{#1}}
\newcommand{\arloopux}[1]{\ar@{.)}@(ur,ul)!U+<.25em,0em>;!U+<-.25em,0em>_{#1}}

\newcommand{\arlooprvx}[1]{\ar@{.)}@(ur,dr)!R+<0em,.25em>;!R+<0em,-.25em>_{#1}}
\newcommand{\arlooplvx}[1]{\ar@{.)}@(ul,dl)!L+<0em,.25em>;!L+<0em,-.25em>^{#1}}
\newcommand{\arloopdvx}[1]{\ar@{.)}@(dl,dr)!D+<-.25em,0em>;!D+<.25em,0em>^{#1}}
\newcommand{\arloopuvx}[1]{\ar@{.)}@(ur,ul)!U+<.25em,0em>;!U+<-.25em,0em>^{#1}}

\newcommand{\scr}[1]{\mathscr{#1}}

\def\C{\mathbb{C}}
\def\E{\scr{E}}
\def\D{\scr{D}}
\def\G{\scr{G}}
\def\I{\scr{I}}
\def\J{\scr{J}}
\def\M{\scr{M}}
\def\Mg{\langle\M\rangle}
\def\MS{\M_\mathrm{S}}
\def\MSg{\langle\M_\mathrm{S}\rangle}
\def\MN{\M_\N}
\def\N{\mathbb{N}}

\def\P{\scr{P}}

\def\R{\mathbb{R}}
\def\S{\scr{S}}
\def\T{\scr{T}}

\def\inv{^{-1}}
\def\tab{\hspace{1em}}

\def\cof{{\vphantom{S}}^{c}\hspace{-0.0833em}}
\def\existsa{{\vphantom{S}}^{\exists}\hspace{-0.0833em}}
\def\nexistsa{{\vphantom{S}}^{\nexists}\hspace{-0.0833em}}
\def\Gm{\varGamma}
\def\1{\mathbf{1}}
\def\2{\mathbf{2}}
\def\n{\bm{n}}
\def\m{\bm{m}}

\def\|{\;|\;}
\def\qA{q_\mathrm{A}}
\def\qR{q_\mathrm{R}}
\def\qe{q_\mathrm{e}}
\def\bs{\text{\textvisiblespace}}
\def\terminal{\Box}
\def\bsigma{\bar{\sigma}}
\def\Qt{\tilde{Q}}
\def\Tt{\tilde{\Theta}}
\def\x{\times}
\def\o{\,\circ\,}
\newcommand{\pii}[1]{{\pi_{#1}}\inv}
\newcommand{\func}[1]{\mathrm{#1}}

\def\dm{\func{dm}}
\def\cdm{\func{cdm}}
\def\id{\func{id}}
\def\cmpl{\func{cmpl}}

\def\add{\func{add}}
\def\sub{\func{sub}}
\def\mult{\func{mult}}
\def\len{\func{len}}

\def\In{\func{in}}
\def\Out{\func{out}}

\def\suc{\func{succ}}
\def\car{\func{car}}
\def\cdr{\func{cdr}}
\def\init{\func{init}}
\def\ltsup{\func{ltsup}}
\def\Inf{\func{Inf}}
\def\lt{\func{lt}}
\def\Max{\func{Max}}

\begin{document}

\newtheorem{definition}{Definition}

\title{\Large Representation and Measure of Structural Information}
\date{}
\author{Hiroshi Ishikawa\\\\
\normalsize Department of Information and Biological Sciences\\
\normalsize  Nagoya City University\\
\normalsize  Nagoya 467-8501, Japan\\
\small hi@nsc.nagoya-cu.ac.jp}
\maketitle

\begin{abstract}
We introduce a uniform representation of general objects that captures the regularities with respect to their structure.
It allows a representation of a general class of objects including geometric patterns and images in a sparse, modular, hierarchical, and recursive manner.
The representation can exploit any computable regularity in objects to compactly describe them, while also being capable of representing random objects as raw data.
A set of rules uniformly dictates the interpretation of the representation
into raw signal, which makes it possible to ask what pattern a given raw signal contains.
Also, it allows simple separation of the information that we wish to ignore from that which we measure,
by using a set of maps to delineate the {\it a priori} parts of the objects,
leaving only the information in the structure.

Using the representation, we introduce a measure of information in general objects relative to structures
defined by the set of maps.
We point out that the common prescription of encoding objects by strings to use Kolmogorov complexity
is meaningless when, as often is the case, the encoding is not specified in any way other than that it exists.
Noting this, we define the measure directly in terms of the structures of the spaces in which the objects reside.
As a result, the measure is defined relative to a set of maps that characterize the structures.
Though it does not depend on Kolmogorov complexity, it raises a question of their relationship,
as the class of applicable objects includes strings.
It turns out that the measure is equivalent to Kolmogorov complexity when it is defined relative to the maps
characterizing the structure of natural numbers.
Thus, the formulation gives the larger class of objects a meaningful measure of information
that generalizes Kolmogorov complexity.
\end{abstract}

\section{Introduction}

What is a pattern?
There does not seem to be a generally accepted mathematical definition.
Intuitively, a pattern is something simpler than it is apparent.
For instance, a repetition of a short substring in a longer string is a pattern:
the longer string is simpler, or contains less information, than most other strings of the same length.
Here, we see a comparison between the apparent size (length in the literal representation) and the ``real'' amount of information.
Formally, this can be stated in terms of Kolmogorov
complexity\cite{Chaitin66,Chaitin69,Kolmogorov65,Solomonoff60,Solomonoff64} of the string,
which is roughly defined as the length of the shortest input to a universal Turing machine that produces the string.
A string can be said to have a pattern if its Kolmogorov complexity is much smaller than its length:
strings that can be effectively described in a significantly shorter description than their length have patterns.
Our goal in this paper is to formalize this notion in the domain of more general objects than strings.

For instance, consider bitmap images.
Ordinary images are much orderly than what is allowed by their representation as an array of colors;
if we take a random bitmap out of all that can be represented as a bitmap, it is almost always a white noise,
rather than what we would consider an ordinary image.
This is similar to the string case where most strings of a given length are random ones that do not have a shorter description than the literal one.
What is the corresponding ``effective description'' of images?
Intuitively, it should be a way to describe the image in which ordinary images can be represented more concisely than noise images.

A Turing machine that produces a bitmap does not suffice because, unlike the case of strings, where all strings
can be represented precisely as they are, the bitmaps are only approximations of what we consider to be real images:
Pixels are artifacts of arbitrary approximation; and we naturally consider bitmaps of various resolutions 
as the same, if they show the same scene.
There would be no problem if it were the case that all important features of an image are independent of the choice
of pixelation.
However, this is clearly not so:
even a notion as simple as that of a line is not so simple to define on bitmaps, especially in such a way a line
in one resolution can be converted into a line in another resolution.

Infinite resolution bitmaps, or functions on an image domain that takes values in the color space,
seem to be good enough for the literal representation.
But then, the objects appearing there are continuous, infinite entities
and thus cannot easily be described effectively as, for instance, an output of a Turing machine.
Yet intuition tells us that some of these infinite entities contain only finite information,
as the extreme cases of ``geometric'' visual patterns shown in Figure~\ref{fig:example1}.

\begin{figure}[t]
\begin{center}
\includegraphics{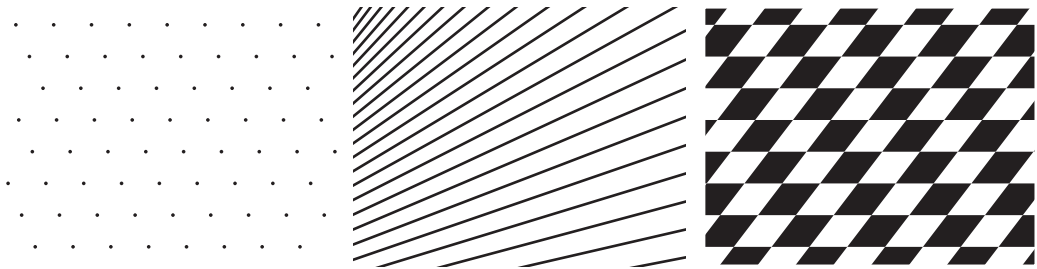}
\caption{Visual patterns.}
\label{fig:example1}
\end{center}
\end{figure}

\subsection{Kolmogorov Complexity Covers All?}\label{kolmocovers}
But surely, one might say, Kolmogorov complexity already covers any domain, since Computer Science teaches us that
information can be encoded by strings.
That is, we can first fix some standard enumeration of the objects, establishing a one-to-one correspondence between
the objects and strings; then we can define the complexity of an object to be the Kolmogorov complexity of the
corresponding string.
That seems to be where such an inquiry usually stops, content with the notion that essentially
we only need to investigate strings.

However, we immediately encounter a few problems.

First, for a class of objects (such as subsets of a Euclidean space) that has a larger cardinality than the set of all
strings, we cannot encode all objects by strings; thus we need to give up the one-to-one correspondence.
We must either encode only some of the objects, encode (perhaps an infinite number of) multiple objects by each string, or employ some combination of the two approaches.
The choice amounts to knowing what to ignore, whether it is some (even most) of the objects that are not encoded,
or the difference between the objects that are encoded into the same string.
{\it How should we make this choice?}

More fundamentally, the resulting measure has little meaning without actually specifying the encoding.
Let $O$ and $S$ be the sets of objects and strings, respectively, and $K(s)$ the Kolmogorov complexity of a string $s$.
With an encoding $f:O\to S$, we might call $K(f(x))$ the complexity of object $x$.
However, if we do not have some good reason to take a particular $f$, we can equally use the encoding
$p\circ f$ with an arbitrary permutation $p:S\to S$.
This observation renders the definition meaningless without $f$ explicitly specified.
So the question is: {\it what is the encoding that gives the complexity some meaning?}
How can we avoid falling into this trap of arbitrariness?
It is certainly not enough just to say that it can be encoded.

With strings, we can choose the identity map as $f$, which gives $K(f(x))$ as much meaning as $K(s)$.
In other cases, however, we need to specify $f$, with at least some justification.
If we insist encoding objects into strings, we need to define a concrete encoding for each class of objects.

Another problem of measuring the information solely through Kolmogorov complexity is that we cannot easily
ignore the part of information we do not care.
For instance, we may try to represent a point in the Euclidean plane by identifying the space with $\R^2$,
i.e., by a pair of real coordinates, and then encoding them by strings.
However, a single real number can contain an arbitrarily large amount of information.
Thus, in this representation, a single point can have an arbitrarily large information when encoded by a string.
That is certainly not what we want here.
Thus, an important part of the encoding is specifying the part of the information we wish to ignore.
But we cannot simply delete such information in the encoding process, since it may be needed to identify and measure
the regularities in the structure later.
If we insist that the computation be carried out strictly in the domain of strings, $f$'s output must contain all
the information in the points.
But after the information has been converted into strings, how do we specify which part of the information should be
ignored?

It would be much better if we can define the notion of computation, such as compression and pattern finding,
directly in terms of the objects we deal with.
What we offer in this paper is a meta definition of $f$ for multiple classes of objects
by specifying how to embed computations in larger spaces.
Central to the formalism is the representation of objects that offers the means to specify the information in individual elements that should be ignored,
while using that very information to find the structures, in which we wish to measure the amount of information.

Thus, the central problem is that of encoding, or representation, of objects.
The paradigm to measure information through computation is the same;
the difference is where the computation takes place.
This question of encoding seems to have suffered a neglect which, in our belief, has prevented us from
formulating a notion of information in objects that have not already been encoded in a convenient way.
We discuss this further in Section \ref{sec:conclusion}.

\subsection{Motivation}
Our motivation for asking this question stems from the desire to model perception.
For perception, there needs to be a large amount of prior knowledge stored in the perceiver, because
perception is an inherently ill-posed problem.
Perception is a process in which the configuration of the signal source is recovered from a signal, as
in recovering a three-dimensional scene from an image.

The problem is that, given the signal, there are usually infinitely many possible source configurations.
Without a preference of possible source configuration on the side of the perceiver, there is no reason
to choose one possibility over another.
For instance, our visual system has a great capability to organize the visual signal into interpretable shapes,
like making sense of the famous Dalmatian photo by R.~C.~James in \cite{Gregory1970}.
To model such a system, it is not enough to know what the possible configurations of the signal source are;
we need to know in advance how likely we are going to encounter each of them.

However, even putting aside the problem of estimating the probabilities, just storing and retrieving the data is
impossible unless we have a very good way to compress the data; for instance, if we store the possible shape of
surfaces as an array of 10 possible heights at each of $10\x 10$ positions, the number of possible surfaces
would be $10^{100}$.
The way this problem has been dealt with is by estimating the probabilities by looking at specific
characteristics of the possible surface.
For instance, the surface smoothness can be computed from a given description of the surface;
we can then decide, for instance, that the smoother the surface is, the higher the probability.
Indeed, the area of computer vision and pattern recognition is full of such heuristics.
Even when machine learning techniques are used, the variables to be learned must be carefully chosen because
we cannot simply learn all possible surfaces.

Our desire to have a measure of information originates from the wish to have a principle for
automatically deciding which quantity to look at and which combinations of variables to learn,
because we consider it reasonable for a perceiving entity to look first for simpler patterns in the signal,
as well as because of the demand of storage efficiency.
That is, if we have a measure of simplicity of general visual objects, we can say, for instance, that the
probability is proportional to the simplicity, or use machine learning techniques to learn
the probabilities that such simpler patterns appear.

In more general terms, this is a problem of inductive inference and modeling: we inductively seek a model of the
world that best explains the data.
There are theories that treat such a problem, and among them are ones with the spirit we describe above.
For instance, the Minimum Description Length (MDL) principle\cite{Rissanen78} advocates Occam's Razor.
Among models that equally fit the data, it chooses the one that is the ``simplest'' in the sense that it allows
for a shorter description of the data.
However, crucially missing from this theory is the problem of representation.
The MDL theory only deals with strings as the data and does not say how the objects should be described by strings.
It is a good principle for people trying to deal with individual problems they understand; but when it comes to
dealing with general objects, it lacks the mathematical concreteness needed to program the principle itself
into machines.

Also, as perceiving entities, we seem to have more interest in the finite part of the data.
One may even say that we can only perceive the finite information out of any infinitely rich source of information,
on the basis  that our capacity of representation is presumably finite.
For instance, if we see a white noise image, we do not perceive the amount of information that can be
encoded in such an image.
Instead, we glean the information that we can; we might just note that it is a white noise,
or if it is a video we would recognize that the noise is constantly changing, and so on.
If we see the three images in Figure~\ref{fig:example2}, which are the same pattern with different
noise added, we do not discriminate among them.
Even though as raw bitmaps they are quite different, we perceive almost nothing about the noise except
for its presence; we just recognize the pattern of the lines as the same and notice that there are some noise.
Thus, to model the perception, we need a way to recognize the part of an infinite signal that represents finite but useful information.
This is why we are especially interested in inherently finite structures whose literal manifestations are infinite.

The human visual system seems to have ``the ability to impose organization on sensory data---to discover
regularity, coherence, continuity, etc., on many levels,'' which is ``apart from both
the perception of tri-dimensionality and from the recognition of familiar objects\cite{WitTen}.''
We agree that such structure and organization that appears at every level is the key to modeling vision and
perception in general.
One purpose of this work is to provide a language to express the perceptual organization that enables us to
implement the ability to impose it on the data.

\begin{figure}[t]
\begin{center}
\includegraphics{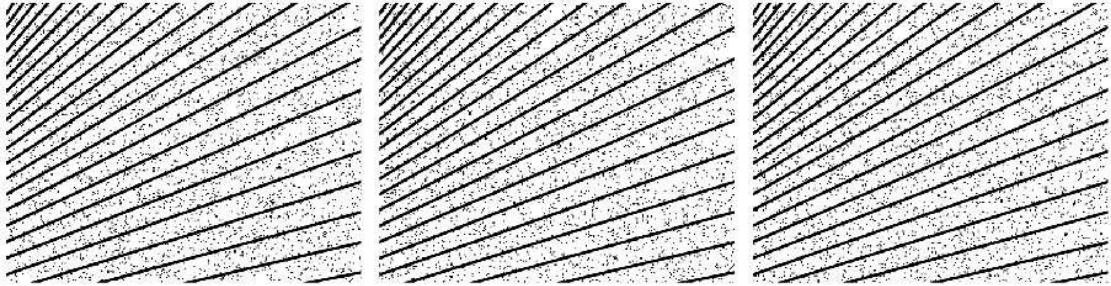}
\caption{Visual patterns with noise.}
\label{fig:example2}
\end{center}
\end{figure}

\subsection{Desiderata}
What we seek is a description, or representation, of general objects with the following properties:
\begin{enumerate}[I.]
\item \label{desideratum:generality}
General: It can represent a general class of objects and all objects in the class, including a part of the descriptions
themselves, allowing hierarchical description.
\item \label{desideratum:uniformity}
Uniform: It represents the objects in a uniform way by simple rules.
\item \label{desideratum:reflectcomplexity}
Reflect complexity: The intuitive complexity, or the amount of information in the object, corresponds to the
complexity of the description.
In particular, intuitively finite object has a finite description.
\item \label{desideratum:groundedness}
Grounded: There is a set of rules that applies to the whole class of objects, not depending on the
instance of the object, dictating how the described objects are related to the raw signal.
\end{enumerate}

In the case of strings, the literal representation satisfies \ref{desideratum:generality}
and \ref{desideratum:uniformity}: representing a string as a string is obviously
general enough to represent any string, and the representation is uniform for any string.
Fixing a universal Turing machine $U$, one can consider a program $p$ for $U$ as a description of a string $x$
if $p$ causes the machine to halt after writing out the string $x$ on the tape.
Intuitively simple string would have a shorter program.
Also, describing strings by other strings automatically satisfy the describability of descriptions and the groundedness.
Thus this description would satisfy all of the desiderata.

In the case of images, we can think of a function on a rectangle in $\R^2$ as the literal representation,
satisfying the desiderata \ref{desideratum:generality} and \ref{desideratum:uniformity}.
But we do not know of a representation that satisfies all of the desiderata.
Perhaps the closest is the page description languages like PostScript, possibly modified to allow infinite precision.
However, it has too many primitives to be convenient for mathematical treatment.
Also, the uniformity and simplicity of the rules of description is important not only for the sake of mathematical
convenience, but also because we aim eventually to develop a way of automatically extract such description from
the literal description, or the signal.
More crucially, PostScript is not general enough: its class of objects is limited to two dimensional pages.
The applicability to more general objects than just images is important because we would like to find
a description that reflects the structure within more abstract data than a two-dimensional page, especially
the description itself.
For instance, if we have a way to describe circles by center points and radii, we would have a three-dimensional
space of circles.
We would like to use the same uniform description to describe a group of the 3D points corresponding
to the circles.
Thus, allowing ``describing the description'' is crucial in order to allow efficient description of repetitive,
hierarchical, and more general structures.

The groundedness requirement is needed to treat general structures in a uniform way.
When we say that some data represents some object, we implicitly assume a set of rules for data interpretation
and manipulation.
It is this set of rules that gives the structure to the object.
It is like a machine with knobs and buttons to control it:
knowing their settings may be enough to determine the state of the machine;
but to describe the effect and
interaction of the machine with the environment and other machines, we need more than the internal parameters.
If the rules are {\it ad hoc}, varying from one instance of representation to another, it would be impossible
to formulate the notion of general structures and describe the manipulation of and interaction between such
structures.

All the desiderata are related to each other.
In particular, we emphasize the following:
it is not enough that the simple objects correspond to less data (\ref{desideratum:reflectcomplexity});
the correspondence must be obtainable from the representation (\ref{desideratum:groundedness})
in a uniform (\ref{desideratum:uniformity}) way that applies to the whole general class (\ref{desideratum:generality}) of objects.
For instance, we can say a pair of a point and a real number represents a circle by regarding them as the
center and the radius.
Or we can say that the pair represents a line by regarding it as a point the line goes through and its angular direction.
But for the representation to cover both cases, one must also include in the representation some data specifying
which case it is for each object.
Such data quickly adds up when one wants to represent various shapes; so when we say that general shapes
are described in a representation, it has not only to cover all the shapes but also include the necessary data in a way
any shape represented can be converted into a common, literal representation.

\subsection{Related Work}

The General Pattern Theory\cite{Gre76,Gre93,GreMil07} is an effort to provide an algebraic framework for
describing patterns as structures.
It defines a vocabulary which is manipulated to cast the concept of pattern in a precise algebraic language.
While it has detailed algebraic and statistical theories with many examples, we only discuss
here the part that deals with the representation of patterns.
The representation is based on graphs.
A graph is fixed; each of its node can be assigned one of {\it generators}, whose set is predetermined;
a restriction as to which combination of generators can be assigned to the nodes is defined as a set of pairwise
restrictions corresponding to the edges of the graph.
There are numerous examples in the literature showing that this representation can be used to represent
many classes of objects.

The Syntactic Pattern Recognition\cite{Fu74} also represents patterns in a way that explicitly handles
the interrelationships between the parts that make up the whole of the object, and use the explicit structure
in patterns to recognize them.
The representation is either by a formal language or a graph.

Neither of the representations used in the two formalisms is satisfactory for us.
The crucial problem is that they are not grounded in the sense above.
They are not uniform from a class of objects to another;
thus, although we can talk about the information in objects in each class, there is no way to compare them
across the classes.
They are general in the sense one can adopt them to many different classes of objects, but they are not
general enough to represent all of the classes in the same uniform way.
They cannot be used, for instance, to define what patterns are, because there is no prescribed way for the representation
to be connected to general enough class of objects.
There is no formal way to give a raw data and ask what pattern it might form.

Besides the Kolmogorov complexity we already mentioned, there are many notions of the complexity of objects.
Most of them are concerned about the complexity of objects that do not include what we deal with in this paper.
Also, note that what we define in this paper is a measure of {\it information} like Kolmogorov complexity and
Shannon information\cite{Shannon48}, rather than a measure of {\it complexity} such as the computational complexity.
We refer the reader to the appendix of \cite{Edmonds99} for an overview of formulations of complexity
with an extensive bibliography.

\subsection{Overview of the Representation}

In this paper, we introduce a representation that fulfills all the desiderata above.
Here, we give an overview of its definition and some of its properties.

We assume that the objects are given {\it a priori} as subsets of some sets.
Those that can be thought of in this way forms a very general class that seems to include most, if not all, objects we
 might deal with.
For example, we can think of a binary string $s=s_0s_1\cdots s_n$ as a subset $\{(i,s_i)\|i=0,\cdots,n\}$ of $\N\times\{0,1\}$;
an image can be thought of as a subset of the product of the image plane and the color space,
i.e., the graph of the image function on the image plane.
A physical object like a bicycle or an automobile, at one level of abstraction, can be thought of as a subset of
$E^3\times M$, where $E^3$ is the 3D Euclidean space and $M$ the set of materials,
e.g., glass, iron, rubber, etc.;
the subset consists of $(x,m(x))$, with $m(x)$ the material that occupies the point $x$ in $E^3$.
We call this representation of objects as subsets the {\it ground representation};
it serves as a signal-level, literal representation.
It is simple to represent something in this way; we can then ask the amount of
information therein.

The ground representation is an abstraction of the kind of data representation that we call the {\it dense}
 representation, which includes strings, bitmaps, and other raw data.
It corresponds to representing a string as itself.
One property of dense representation is that the presence of regularities does not affect it.
For instance, any image can be represented as a bitmap in exactly the same manner, whether it is
a regular image or a white noise.
Another type of representation, which we call the {\it sparse} representation, utilizes regularities
in the object to describe it.
In the image example, if it is an image of geometric objects, it should take very little data
to describe it, at least in principle;
a circle, for instance, can be represented just by specifying its center point and radius.
The same kind of description cannot be used to describe a white noise.

An important feature of the representation proposed in this paper is that it can interpolate between the dense
and sparse representation, so that it can take advantage of regularity in the data while also being capable
of representing any, even random, data.
This is similar to using an input to a Turing machine to represent strings.

As the vocabulary to describe the regularities in such objects, we use the maps that characterize
the space in which the objects are included as subsets.
Maps characterize the structure of spaces in the following sense.
Any two sets with the same cardinality are the same sets in absence of other characteristics.
For instance, $\R$ is the same set in this sense as $\R^2$, i.e., there exists a one-to-one map between them,
if we disregard the structures such as the topology, the vector
space structure, the metric structure, the order, and the algebraic structure.
These structures can be characterized by maps.
For instance, the metric structure is defined by the distance map that gives the distance between two elements of the set;
the order is given by a predicate on a pair of elements that returns {\it true} if the first element is less than the second.
The two sets are different when we consider the structures because the one-to-one map does not commute with
the maps that define the structures.

Using such {\it structure maps} to describe regularity, objects with regularities can be represented
through sparse parameters in the proposed representation.
For a given object, there can be many different ways of representing it, just as there can be
any number of inputs to a Turing machine producing the same string.
Importantly, there is a prescribed way to connect the description to the ground representation;
thus, the representation is grounded.
While taking structures into account so that regular objects can be represented as such,
it automatically provides an interpretation of each represented object into the signal level.
In other words, the relationship between the parameters and the data is part of the representation.
Thus, we can give our data in the ground representation and then ask what sparser,
more structured representation is possible.

Let us be slightly more concrete.
In the proposed representation, we take a number of sets and maps between them (that are to be composed
by the structure maps), which we call a {\it diagram}.
Then we call an assignment of a subset to each set in the diagram its {\it cross section}.
A cross section must satisfy a certain constraints, because of which
we can uniquely determine all the subsets by specifying only a partial cross section, which assigns
subsets to only some of the sets in the diagram.
If one of the subsets coincides with the ground representation of the object in question,
we say that it is represented by the diagram and the partial cross section.

The maps define the structures we take into consideration,
which determine the regularities, which in turn allow more concise description of the object than the literal one.
Since the representation by diagrams and cross sections is explicitly in terms of the maps in the diagram,
it is apparent from the diagram exactly what structure is taken into account.

Some of the properties of the representation are as follows. It is:
\begin{enumerate}[i)]
\item Sparse: Unlike dense representations such as bitmap, it is capable of representing
objects by a combination of their essential structure and instance-specific parameters.
The diagram expresses the essential structure while the partial cross section represents the parameters.
Implementation-dependent approximation of representation only affects the parameters and thus
can be separated from the structure.
The sparseness also makes it  flexible and easy to manipulate.
By modifying the parameters, different instances of the same structure can be easily represented.
Also, comparison of two patterns having the same structure is naturally defined.

\item Modular:
Parts of the representation can be understood as the modules to construct larger and more complex ones.
Complex combinations can be obtained hierarchically and recursively as well as by simple union and intersection.

\item Hierarchical:
Because it can be applied to any data, it can also be applied to the parameter space
parametrizing some other structure, leading to a hierarchical representation.

\item Recursive: It can represent a recursively defined structure,
making it particularly powerful in, for instance, representing repeated patterns.
The ``repeat'' can be in various spaces that can manifest in the final pattern in non-obvious ways.
\end{enumerate}

Finally, diagrams and their cross sections can represent maps between powersets.
In fact, the representation of subsets can be considered a special case where the map sends a trivial set to the subset.
Any computation, in particular, can also be represented.

\subsection{Measure of Structural Information}\label{measurestrinfo}

Using the representation, we introduce a measure of information.
Roughly speaking, it is defined as the size of the smallest diagram representing the object,
where the diagram can contain only those maps that can be composed by a set of given structure maps,
including constant maps.

Thus, the measure is relative to the structure and constants expressed explicitly in the form of maps.
The explicit incorporation of the structure of the object space is the key to avoiding the trap of arbitrariness.
The patterns such as shown in Figure~\ref{fig:example1} all have finite information according to the measure.
The measure is relative to the constants because of the aforementioned need to separate the information in the structure from that in infinite objects such as real numbers.

Because of the reasons laid out in \ref{kolmocovers},
we do not follow the recipe of interfacing Kolmogorov complexity by encoding objects by
strings; strings are given no special status in this theory.
Instead, we define it directly in terms of the structure of the spaces in which the objects reside.
As such, the new measure does not depend on Kolmogorov complexity.
Since the class of applicable objects includes that of strings, however, it raises a question of their relationship.
It turns out that the new measure is equivalent to Kolmogorov complexity in the case where strings are characterized by the structure of natural numbers given by the constant $0$ and the successor function.
Note that this is not obvious {\it a priori}:
the definition of the representation and the measure does not even mention strings.
Also, the new measure is defined relative to structure maps.
It is equivalent to Kolmogorov complexity when it is defined relative to this particular set of structure maps;
relative to other sets, it may not be.
If the set includes the constant maps of all strings, for example, any string's information would be 1.

Thus, the new measure gives the larger class of objects a meaningful measure of information
that generalizes Kolmogorov complexity.

\vspace{1em}
The rest of the paper is organized as follows.
In the next section, we define the notion of diagrams and their cross sections precisely,
as well as what is meant by representing with them.
We also list the notations used throughout this paper.
In section \ref{sec:geopat}, we illustrate some properties of the representation by geometric examples.
In section \ref{sec:computation}, we give more examples, this time those representing computations.
In section \ref{sec:measure}, we define the information measure of structure of general objects.
In section \ref{sec:RelKolm}, we prove that the measure generalizes Kolmogorov complexity.
In section \ref{sec:conclusion}, we further discuss the difference between our approach and the string-centered one,
before concluding.

\section{Representation by Diagrams and Cross Sections}

\subsection{Definitions}\label{subsec:def}

We fix the notation for standard finite sets as $\1=\{0\}, \2=\{0,1\}, \cdots, \n=\{0,1,\cdots,n-1\}$, etc.
The set $\2$ is also used as the set of Boolean values, $0$ meaning {\it false} and $1$ {\it true}.
We mean by $f: X\to Y$ that $f$ is a map from $X$ to $Y$.
We denote the set of all subsets of $X$ (the power set of $X$) by $\2^X$.
The map from $\1$ to $X$ that maps $0\in\1$ to $x\in X$ is denoted by the same letter $x:\1\to X$.
We call it a {\it constant map}.

\begin{defn}
Let $\S=(S_i )_{i \in \I}$ be a family of sets indexed by a set $\I$.
A {\it cross section} $s$ of $\S$ is an assignment to each
set $S_i$ in  $\S$ of its subset $s_i\subset S_i$.
\end{defn}

In other words, a cross section of $\S$ is another family $(s_i )_{i \in \I}$ of sets indexed by $\I$
such that $s_i\subset S_i$ for all $i\in\I$.
We used the index set to make clear that there can be multiple members of the family $\S$ that are identical as sets;
however, we avoid the use of indices almost entirely in this paper.
We use the set-theoretic notation with $\S$ such as $S\in\S$.
The equality of two members of $\S$ means that their indices are the same;
if the indices are different, we treat them as different, even if they are identical as sets.
When we discuss a set $S=S_i$ in $\S$ and a cross section $s$ of $\S$,
$s(S)$ denotes the subset $s_i$ assigned to $S_i$ by $s$.
Thus, $s$ assigns each $S\in \S$ its subset $s(S)\subset S$.

We denote the set of cross sections of $\S$ by $\Gm(\S)$.
Let  $\T$ be a subfamily of $\S$.
A cross section of $\T$ is called a partial cross section of $\S$.
For a cross section $s$ of  $\S$, the cross section of $\T$ that assigns $s(S)$ to $S$ in $\T$ is 
called the restriction of $s$ to $\T$, denoted by $s|_\T$.
For a cross section $t$ of $\T$, we denote the set $\{s\in\Gm(\S)\,|\, s|_\T=t\}$
of cross sections of $\S$ that restrict to $t$ by $\Gm(\S\,|\,t)$.

\begin{defn}
A {\it diagram} is a triple $(\S, \S',\M)$ of a family $\S=(S_i )_{i \in \I}$ of sets, its subfamily $\S'$, and
a family $\M=(\varphi_j )_{j\in\J}$ of maps of the form $\varphi_j:\2^S\to\2^T$, with $S, T\in\S$,
where $\I, \J$ are index sets.
\end{defn}

A diagram $(\S, \S', \M)$ such that both $\S$ and $\M$ are finite is called a finite diagram.
Let $(\S, \S',\M)$ be a diagram.
There are maps $\dm:\M\to\S$ and $\cdm:\M\to\S$ such that
$\varphi:\2^{\dm(\varphi)}\to\2^{\cdm(\varphi)}$ for $\varphi\in\M$.
Also, define the maps $\Out=\dm\inv:\S\to\2^\M$ and $\In=\cdm\inv:\S\to\2^\M$ so that
$\Out(S)=\{\varphi\in\M\,|\,\dm(\varphi)=S\}$ and
$\In(S)=\{\varphi\in\M\,|\,\cdm(\varphi)=S\}$ for $S\in\S$.
\begin{defn}\label{def:crosssection}
The {\it cross section of a diagram} $(\S,\S',\M)$ is a cross section $s$ of $\S$ such that,
for any $S\in\S$ with $\In(S)\neq\emptyset$, the following holds:
\begin{align}
s(S)&=\bigcap_{\varphi\,\in\,\In(S)}\varphi(s(\dm(\varphi)))\hspace{.5em}  \text{if } S\in\S\setminus\S',\label{eq:section}\\
s(S)&=\bigcup_{\varphi\,\in\,\In(S)}\varphi(s(\dm(\varphi)))\hspace{.5em}  \text{if } S\in\S'.\label{eq:section2}
\end{align}
\end{defn}

In diagram $(\S, \S', \M)$, the subfamily $\S'$ of $\S$ specifies the sets for which a cross section should satisfy
\eqref{eq:section2} instead of \eqref{eq:section};
this means that the cross section on that set should be the union, rather than the intersection, of the images by the
incoming maps.
We denote the set of cross sections of diagram $(\S,\S',\M)$ by $\Gm(\S,\S',\M)$.
We also define $\Gm(\S,\S',\M\,|\,t)=\Gm(\S,\S',\M)\cap\Gm(\S\,|\,t)$ for $\T\subset\S$ and $t\in\Gm(\T)$;
i.e., $\Gm(\S,\S',\M\,|\,t)$ is the set of cross sections of diagram $(\S,\S',\M)$ that restrict to the
cross section $t$ of subfamily $\T$ of $\S$.

To illustrate the definitions by example, suppose $\S=\{\1,X, Y, Z,W\}$,
$\S'=\{W\}$, $w\in W$, and $\M=\{w,\varphi,\psi,\eta,\delta,\kappa\}$ with
\begin{align*}
&w:\2^\1\to \2^W,\hspace{2em}\varphi: \2^Y\to \2^X,\hspace{2em}\psi: \2^Z\to \2^Y,\\
&\eta: \2^W\to \2^Y,\hspace{2em}\delta: \2^W\to \2^X,\hspace{2em}\kappa: \2^X\to \2^W,
\end{align*}
where $w$ is an element of $W$; the same letter denotes a constant map.
We denote the diagram $(\S,\S',\M)$ as follows:
\begin{equation}
\diagram[@C=4em]{
{X}_{(1)} \aux[@<.3em>]{rd}{\kappa} & {Y}_{(2)} \av{l}{\varphi} & {Z}_{(3)}  \av{l}{\psi} \\
 & {W}_{(4)} \ajurx{w} \au{u}{\eta} \au[@<.3em>]{ul}{\delta}
}\label{samplediagram}
\end{equation}
For instance, $\varphi$ maps each subset of $Y$ to a subset of $X$.
We omit the set $\1$ from the diagram: a constant map is shown as an incoming arrow, without the domain $\1$.
Note also that the arrows have dotted shafts, which signifies that the map is between power sets.
The parenthesized subscript numbers are for reference:
as more than one sets in the family can be identical as sets, we use these to refer to them.
We always use $S_i$ to mean the set with the subscript $(i)$ in the diagram under discussion.
Thus, if we are discussing the one in equation \eqref{samplediagram}, $S_1$ means the set $X$ and
$S_2$ means $Y$, etc.
Also, there are two kinds of arrowheads: the ordinary arrows and round ones.
An arrow has the round arrowhead if and only if it is coming into a set in $\S'$, in this case $S_4=W$.
For a cross section $s$ of this diagram, we have, e.g.,
\[
\In(X)=\{\varphi,\delta\},\hspace{1.2em}
\dm(\varphi)=Y,\hspace{1.2em}
s(X)=\varphi(s(Y))\cap\delta(s(W)),\hspace{1.2em}
s(W)=\kappa(s(X))\cup\{w\}.
\]

A diagram and its partial cross section represents an object in the following sense.
\begin{defn}
Let $(\S,\S',\M)$ be a diagram, $X$ a set in $\S$, $\T$ a subfamily of $\S$, and $t$ a cross section of $\T$.
Suppose an object is represented in the ground representation as a subset $A$ of $X$.
The object is said to be represented by $(\S, \S', \M, \T, t, X)$ if $s(X)=A$ for every cross section
$s$ in $(\S, \S', \M\,|\,t)$.
\end{defn}

The ground representation is a special case of this as a {\it trivial representation};
just take $\S=\S'=\T=\{X\}, \M=\emptyset$, and $t(X)=A$.
Thus the representation is general enough to include all dense representation.
The aim, however, is to enable more efficient representation that captures the structure.

Here, we also define the concepts of minimality and limit for later use.
\begin{defn}
Let $(\S,\S',\M)$ be a diagram, $X\in\S$, and $G\subset\Gm(\S,\S',\M)$.
A cross section $s\in G$ such that no other $t\in G$ gives $t(X)\subsetneq s(X)$ is said to be {\it minimal}
 on $X$ in $G$. We denote the set of such cross sections by $\min_{X}G$.
\end{defn}

Note that $\min_{Y}\min_{X}G\supset\min_{X}G\cap\min_{Y}G$ for $X,Y\in \S$:
if $s\in\min_{X}G\cap\min_{Y}G$, then $s(Y)\subset t(Y)$ for any $t\in\min_{X}G$ since $t\in G$ and $s\in\min_{Y}G$;
thus $s\in\min_{Y}\min_{X}G$.
Since it is also the case that $\min_{Y}\min_{X}G\subset\min_{X}G$, by symmetry it follows that
$\min_{Y}\min_{X}G\cap\min_{X}\min_{Y}G=\min_{X}G\cap\min_{Y}G$.

\begin{defn}
Let $(\S,\S',\M)$ be a diagram, $X\in\S, \T\subset\S$, and $t\in\Gm(\T)$.
Furthermore, let $X_1, \cdots, X_n$ be a finite number of sets in $\S$.
A subset $A$ of $X$ is said to be {\it represented by the data $(\S,\S',\M, \T, t, X, (X_1, \cdots, X_n))$ as a limit}
if $s(X)=A$ for any cross section $s$ in $\min_{X_n}\min_{X_{n-1}}\cdots\min_{X_1}\Gm(\S, \S', \M\,|\,t)$.
\end{defn}

\subsection{Notations}\label{notations}
Here we list some more notations used in this paper.

\begin{enumerate}[i)]

\item \label{setmaps} For any set $X$, $\id:X\to X$ denotes the identity map on $X$ and
$\omega:X\to \1$ the unique map from $X$ to $\1=\{0\}$.
The complement map $\cmpl:\2^X \to\2^X$ is defined for $A\subset X$ by:
\begin{equation}
\cmpl(A)=\cof A= X\setminus A.\label{def:cmpl}
\end{equation}

\item The product map $f_1\x\cdots\x f_n: X\to Y_1\x\cdots\x Y_n$ of maps
$f_i:X\to Y_i\;(i=1,\cdots, n)$ is defined by $(f_1\x\cdots\x f_n)(x)=(f_1(x),\cdots, f_n(x))$.
Given a map $f: X\to Y$ and a constant map $z:\1\to Z$, one can construct a product map:
\[
f\x (z\circ\omega):X\to Y\x Z
\]
of $f$ and $z\circ\omega:X\to Z$.
By abuse of notation, we denote this map by $f\x z:X\to Y\x Z$.
Similarly, we mix maps of the form $X\to Y$ and $\1\to Z$ freely in making a product map.

\item For a Cartesian product $X_1\x X_2\x\cdots\x X_n$, the map
\[
\pi_i : X_1\x X_2\x\cdots\x X_n \to X_i
\]
is the projection to the $i$'th component.
We use a shorthand $\pi_{ij}$ for the product map
$\pi_i\x\pi_j:X_1\x\cdots\x X_n \to X_i\x X_j$,
$\pi_{ijk}$ for  $\pi_i\x\pi_j\x\pi_k$, and so on.

\item For a disjoint union $X_1+ X_2 + \cdots + X_n$, the map
\[
\iota_i : X_i\to X_1 + X_2 + \cdots + X_n
\]
is the injection from the $i$'th component.

\item The map union $f+g: X+Y \to Z$ of maps $f:X\to Z$ and $g:Y\to Z$ is defined by $(f+g)(x)=f(x)$ if $x\in X$ and $(f+g)(x)=g(x)$ if $x\in Y$.

\item \label{powermapitem} For a map $f:X\to Y$, we denote by the same letter the map
$f:\2^X\to\2^Y$ between the power sets defined by $f(A)=\{f(x)\,|\,x\in A\}\subset Y$
for $A\subset X$.

\item For a map $f:X\to Y$, the map $f^{-1}:\2^Y\to\2^X$ is defined by
$f\inv(A)=\{x\in X\,|\,f(x)\in A\}\subset X$ for $A\subset Y$.
By a slight abuse of notation, by $f\inv(y)$ for $y\in Y$ we mean $f\inv(\{y\})$.

\item For a map $f:X\to X$, $f^0$ denotes $\id_X$.
For a positive integer $n$, $f^n$ denotes the map $f^n:X\to X$ defined as applying $f$ for $n$ times
as well as the map $f^n:\2^X\to\2^X$ defined as in \ref{powermapitem}).
When $n$ is a negative integer, $f^n$ denotes the map $f^n:\2^X\to\2^X$ defined as applying
$f^{-1}:\2^X\to\2^X$ for $-n$ times.

\end{enumerate}

\section{Geometric Patterns}\label{sec:geopat}

Using diagrams and cross sections, we can represent geometric objects in a uniform and compact way.
In this section, we introduce the representation and discuss its properties using examples.

\subsection{Examples}\label{visexample}

As the simplest example, we consider a circle in the Euclidean plane $X$.
Let us denote the vector space of translations in $X$ by $V$.
Also, denote the map that sends $(x,y)\in X\x X$ to $x-y$ by $\sub: X\x X \to V$, and
the map that gives the length of a vector by $\len: V\to \R$.

Consider the following diagram:
\begin{equation}
\diagram{
{\R}_{(1)} \au{r}{\len\inv} & {V}_{(2)} \au{d}{\sub\inv} \\
{X}_{(3)} \aul{r}{.4}{\pii{2}} & {X\x X}_{(4)} \aul{r}{.55}{\pi_1} & X_{(5)}
}
\label{dia:circle}
\end{equation}
This denotes a diagram $(\S,\S',\M)$ with
\[
\S=(S_1, \cdots, S_5),\; \S'=\emptyset,\; S_1=\R,\; S_2=V,\; S_3=S_5=X,\; S_4=X\x X
\]
 and
\[
\M =(\len\inv:\2^{S_1}\to\2^{S_2}, \sub\inv:\2^{S_2}\to\2^{S_4}, \pii{2}:\2^{S_3}\to\2^{S_4}, \pi_1:\2^{S_4}\to\2^{S_5}).
\]

Note that, while the inverse maps are indicated by $\inv$, the power map in the forward direction is to be surmised
from the convention that the map is between powersets.

Suppose that $\T=\{S_1, S_3\}$ and that its cross section $t$ is defined by
\[
t(S_1)=\{r\},\;\;\;t(S_3)=\{p\},
\]
where $r$ is a positive real number and
$p$ is a point in the Euclidean plane $X$.
Let $s$ be a cross section in $\Gm(\S,\S',\M\,|\,t)$.
Then, by \eqref{eq:section}, we have
\begin{align*}
s(S_2)&=\len\inv(s(S_1))=\len\inv(t(S_1))=\len\inv(\{r\})=\{v\in V\|\len(v)=r\},\\
s(S_4)&=\sub\inv(s(S_2))\cap\pii{2}(s(S_3))\\
&=\{(x,y)\in X\x X\|x-y\in s(S_2), y\in s(S_3)\},\\
s(S_5)&=\{\pi_1((x,y))\in X\|(x,y)\in s(S_4)\}\\
&=\{x\in X\| x-y\in s(S_2), y\in s(S_3)\}\\
&=\{x\in X\| \len(x-p)=r\}.
\end{align*}
Thus the cross section $s$ is completely determined and $s(S_5)$ is the set of the points on the circle centered at $p$ with radius $r$.
In this way, $(\S,\S',\M, \T, t, S_5)$ represents the circle.

If $t(S_3)=\{p, q\}$, it represents two circles with the same radius $r$ centered at $p$ and $q$.
Thus, we can think of $S_3$ as the space of centers of the circles.
If $t(S_1)=\{r, t\}$ instead, it would represent two concentric circles with radii $r$ and $t$.
If we modify the diagram to
\[
\diagram[@C=5em]{
{\R\x X}_{(1)} \au{r}{((\len\o\pi_1)\x\pi_2)\inv} & {V\x X}_{(2)} \au{r}{(\sub\x\pi_2)\inv}
 & {X\x X}_{(3)} \aul{r}{.55}{\pi_1} & X_{(4)}
}
\]
and let $\T=\{S_1\}$ and define $t\in\Gm(\T)$ by $t(S_1)=\{(r_1,p_1),(r_2,p_2),\cdots\}$,
then we have
\begin{align*}
s(S_2)&=\{(v,x)\in V\x X\|(\len(v),x)\in t(S_1)\},\\
s(S_3)&=\{(x,y)\in X\x X\|(x-y,y)\in s(S_2)\},\\
s(S_4)&=\{x\in X\| \existsa y\in X, (\len(x-y),y)\in t(S_1)\},
\end{align*}
and we have as $s(S_4)$ the circles specified by the radius-center pairs in $t(S_1)$.

For another example, a line in $X$ can be represented using the following diagram:
\begin{equation}
\diagram{
V_{(1)} \aul{r}{.4}{\pii{1}} & {V\x\R}_{(2)} \aul{d}{.4}{\sub\inv\o\mult} \\
X_{(3)} \aul{r}{.4}{\pii{2}}  & {X\x X}_{(4)}\au{r}{\pi_1} & X_{(5)}
}\label{linediagram}
\end{equation}
Here, $\mult: V\x \R \to V$ is the scalar multiplication.
Other maps are as above.
Suppose that $\T=\{S_1, S_3\}$ and that its cross section $t$ is defined by
\[
t(S_1)=\{v\},\;\;\;t(S_3)=\{p\},
\]
where $p$ is a point in the Euclidean plane $X$ and $v$ is a vector in $V$.
Let $s$ be a cross section in $\Gm(\S,\S',\M\,|\,t)$.
Then, from \eqref{eq:section} we have
\begin{align*}
s(S_2)&=\pii{1}(\{v\})=\{(v,c)\in V\x\R\|c\in\R\},\\
s(S_4)&=\sub\inv(\mult(s(S_2)))\cap\pii{2}(\{p\})\\
&=\sub\inv(\{cv\in V\|c\in\R\})\cap\pii{2}(\{p\})\\
&=\{(x,p)\in X\x X\|\existsa c\in\R, x-p=cv\},\\
s(S_5)&=\{x\in X\|\existsa c\in\R, x-p=cv\}\\
&=\{p+cv\|c\in\R\}.
\end{align*}
Thus, the cross section $s$ is completely determined and $s(S_5)$ consists of the points on the line that goes through $p$
and has the direction parallel to $v$.

\subsection{Union}\label{sec:union}

As mentioned in \ref{subsec:def},
\[
\diagram{S_1\au{r}{\phi} & S_2  & S_3 \av{l}{\psi}}
\]
denotes the case when \eqref{eq:section} in Definition~\ref{def:crosssection} is required,
i.e., $S_2\in\S\setminus\S'$.
Any cross section $s$ of the diagram satisfies $s(S_2)=\phi(s(S_1))\cap \psi(s(S_3))$.
To denote the other case, we use
\begin{equation}
\diagram{
S_1\aux{r}{\phi} & S_2 & S_3 \avx{l}{\psi}
}\label{unionorg}
\end{equation}
to indicate that $S_2\in\S'$ and $s(S_2)=\phi(s(S_1))\cup \psi(s(S_3))$.
Thus, for any set in $\S'$, incoming maps are depicted with the round arrow.

In the examples, we may use two kinds of incoming arrows as:
\[
\diagram{
S_1\aux{r}{\varphi} & S_2 & S_3 \avx{l}{\psi} \\
S_4\aux{ru}{\theta} &    & S_5\av{lu}{\eta}
}
\]
It means $s(S_2)=(\varphi(s(S_1))\cup \psi(s(S_3))\cup \theta(s(S_4)))\cap \eta(s(S_5))$,
i.e., we take the unions first, and then the intersection.
This is simply an abbreviation of
\[
\diagram[@C=4em]{
S_1\aux{r}{\varphi} & S_6 \au{d}{\id} & S_3 \avx{l}{\psi} \\
S_4\aux{ru}{\theta} &  S_2  & S_5\av{l}{\eta}
}
\]

If we allow the complement map $\cmpl:\2^X\to\2^X$ in diagrams, we need only one of the conditions
in Definition~\ref{def:crosssection}, because we can make unions from intersections or
vice versa.
Using $\cmpl$,
\begin{equation}
\diagram[@C=4em]{
S_1\au{r}{\cmpl\o\phi} & S_4 \au{d}{\cmpl}  & S_3 \av{l}{\cmpl\o\psi}\\
 & S_2
}\label{unioncmpl}
\end{equation}
\begin{align*}
s(S_4)&=\cof\phi(s(S_1))\cap \cof\psi(s(S_3))\\
s(S_2)&=\cof s(S_4) = \phi(s(S_1))\cup \psi(s(S_3)).
\end{align*}
Thus, \eqref{unioncmpl} is equivalent to \eqref{unionorg}.

\subsection{Representing maps}
A diagram with partial cross section can represent a map in the following sense:
\begin{defn}
A map $\varphi:\2^X\to\2^Y$ is said to be represented by $(\S,\S',\M, \T, t, X, Y)$ if $(\S,\S',\M)$ is
a diagram, $\T\subset\S$, $t\in\Gm(\T)$, $X,Y\in\S$, and every cross section $s$ in $\Gm(\S,\S',\M\|t)$
satisfies $s(Y)=\varphi(s(X))$.
\end{defn}

As an example, let us represent the map $\mathrm{cl}:\2^X\to\2^X$ that maps a subset $A$ of a Euclidean space $X$
to the topological closure $\bar{A}$ of $A$ in $X$.
Consider the diagram $(\S,\S',\M)$:
\[
\diagram[@C=4.3em]{
X_{(1)}\aul{r}{.35}{\pii{1}} & {X\x X}_{(2)}\au{r}{\sub\x\pi_2} &
{V\x X}_{(3)}\au{r}{(\len\o\pi_1)\x\pi_2} & {\R\x X}_{(4)} \au{r}{\Inf} & {\R\x X}_{(5)} \au{d}{\pi_2}\\
& & & {\R}_{(6)} \au{ru}{\pii{1}}& X_{(7)}
}
\]
Here, the map $\Inf$ maps $B\subset\R\x X$ to
$\{(a, x)\in\R\x X\| x\in \pi_2(B), a=\inf \pi_1(B\cap \pii{2}(x))\}$;
thus, for each $x\in X$ that appears in $B$, there is an element $(a,x)$ in $\Inf(B)$,
where $a$ is the infinimum of the set of real numbers $b$ that appear as $(b,x)$ in $B$.
Now, if $\T=\{S_6\}, t(S_6)=\{0\}, s\in\Gm(\S,\S',\M\|\T)$, and $s(S_1)=A$, we have
\begin{align*}
s(S_2)&=\{(x,y)\|x\in A, y\in X\}\\
s(S_4)&=\{(d,y)\|y\in X, \existsa x\in A, \len(x-y)=d\}\\
s(S_7)&=\{y\in X\|\inf_{x\in A} \len(x-y)=0\}=\bar{A}.
\end{align*}
Thus $(\S,\S',\M, \T, t, S_1, S_7)$ represents $\mathrm{cl}$.

The infinimum map $\Inf$ in turn can be represented by
\[
\diagram[@C=4.5em]{
{\R\x X}_{(1)}\aul{r}{.4}{\pii{13}} & {\R\x\R\x X}_{(2)}\au{r}{\pi_2\x (\lt\o\pi_{12}) \x\pi_3} &
{\R\x\2\x X}_{(3)} \au{r}{\cmpl\o\pi_{13}} & {\R\x X}_{(4)}\au{d}{\Max}\\
& & {\2}_{(5)} \au{u}{\pii{2}} & {\R\x X}_{(6)}
}
\]
with $\T=\{S_5\}, t(S_5)=\{1\}$.
The map $\lt:\R\x\R\to\2$ maps $(a,b)$ to $1$ if $a<b$ and $0$ otherwise, while
the map $\Max$ maps $B\subset\R\x X$ to
$\{(a, x)\in \R\x X\| x\in \pi_2(B), a=\max \pi_1(B\cap \pii{2}(x))\}$.
Then if $s\in\Gm(\S,\S',\M\|\T)$ and $s(S_1)=B$ we have
\begin{align*}
s(S_2)&=\{(a,b,x)\|(a,x)\in B, b\in\R\},\\
s(S_3)&=\{(b,1,x)\|\existsa (a,x)\in B, a< b\},\\
s(S_4)&=\{(b,x)\|\nexistsa (a,x)\in B, a< b\},\\
s(S_6)&=\{(c, x)\| x\in \pi_2(B), c=\max \{b\in\R\|\nexistsa (a,x)\in B, a<b\}\}\\
&=\{(c, x)\| x\in \pi_2(B), c=\inf (B\cap \pii{2}(x))\}.
\end{align*}
Thus $(\S,\S',\M, \T, t, S_1, S_6)$ represents $\Inf$.

Finally, the maximum map $\Max$ can be represented by
\[
\diagram[@C=4.5em]{
{\R\x X}_{(1)} \aul[@<.1667em>]{d}{.4}{\pii{23}} \avl[@<-.1667em>]{d}{.4}{\pii{13}} \au{r}{\id}
& {\R\x X}_{(2)}\\
{\R\x\R\x X}_{(3)}\au{r}{\pi_1\x (\lt\o\pi_{12}) \x\pi_3} &
{\R\x\2\x X}_{(4)} \av{u}{\cmpl\o\pi_{13}}
& {\2}_{(5)} \avl{l}{.3}{\pii{2}}
}
\]
with $\T=\{S_5\}, t(S_5)=\{1\}$.
Then if $s\in\Gm(\S,\S',\M\|\T)$ and $s(S_1)=B$ we have
\begin{align*}
s(S_3)&=\{(a,b,x)\|(a,x), (b,x)\in B\},\\
s(S_4)&=\{(a, 1, x)\| (a,x)\in B, \existsa (b,x)\in B, a<b\},\\
s(S_2)&=\{(a, x)\in B\| \nexistsa (b,x)\in B, a<b\}.
\end{align*}
Thus $(\S,\S',\M, \T, t, S_1, S_2)$ represents $\Max$.

\subsection{Recursive Definition}
Consider the following diagram:
\begin{equation}
\diagram{
 V_{(1)}\aul{r}{.4}{\pii{2}} & {X\x V}_{(2)}\aux[@<.4ex>]{d}{\add}\\
X_{(3)} \aux{r}{\id} & X_{(4)} \au[@<.6ex>]{u}{\pii{1}}
}\label{dotsdiagram}
\end{equation}
Here, $\add: X\x V\to X$ is the parallel translation
$(x,w)\mapsto x+w$ in the Euclidean space $X$.

Suppose that $\T=\{S_1, S_3\}$ and that its cross section $t$ is defined by
\begin{equation}
t(S_1)=\{v\},\;\;\;t(S_3)=\{p\}, \label{sec:dots}
\end{equation}
where $p$ is a point in the Euclidean plane $X$ and $v$ is a vector in $V$.
Let $s$ be a cross section in $\Gm(\S,\S',\M\,|\,t)$.
Then, from \eqref{eq:section} we have
\begin{align}
s(S_2)&=\{(x,w)\| x\in s(S_4), w\in s(S_1)\},\nonumber\\
s(S_4)&=\{p\}\cup\{x+w\| (x,w)\in s(S_2)\}\nonumber\\
&=\{p\}\cup\{x+w\| x\in s(S_4), w\in s(S_1)\}\label{sec:dotsS4}
\end{align}

From \eqref{sec:dotsS4}, clearly $s(S_4)\supset D=\{p$, $p+v$, $p+2v$, $p+3v$, $\cdots\}$, i.e.,
$s(S_4)$ contains equally spaced points beginning at $p$ and separated by $v$.
However, this does not uniquely determine the cross section:
for instance, we can define $s(S_4)=X$;
or indeed any set that is the union of $D$ and a set invariant under the
translation by $v$.

To make it unique, we can take $\min_{S_4}\Gm(\S,\S',\M\,|\,t)$.
Then it only contains the cross section with $s(S_4)=D$.

Or we can use the following proposition.
Let $\N=\{0,1,2,3,\cdots\}$ denote the set of natural numbers.
\begin{prop}\label{grading}
Suppose that a set $S$ has a ``grading'' function $g:S\to\N$ and let $S_n$ denote $g\inv(n)$ for $n\in\N$.
Consider a map $\eta:\2^S\to \2^S$ that satisfies, for $i\in\N$,
\[
\eta(S_i)\subset S_{i+1},\hspace{2em} \eta(S)=\bigcup_{n=0}^\infty \eta(S_n).
\]
If $S$ can be written $S=S_0\cup \eta(S)$, then
\[
S=\bigcup_{n=0}^\infty \eta^n(S_0).
\]
\end{prop}
\begin{proof}
Since $S_n\neq S_m$ if $n\neq m$,
for $x\in S_{n+1}$ with $n\in\N$, $x\notin S_0$ and $x\notin \eta(S_m)\subset S_{m+1}$ if $m\neq n$.
Thus $x\in\eta(S_n)$ follows from
\[
S=S_0\cup \eta(S)=S_0\cup\bigcup_{n=0}^\infty \eta(S_n).
\]
Thus $S_{n+1}\subset\eta(S_n)$.
Since $\eta(S_n)\subset S_{n+1}$, it follows $\eta(S_n)=S_{n+1}$.
Therefore, $S_n=\eta(S_{n-1})=\eta(\eta(S_{n-2}))=\cdots=\eta^n(S_0)$.
The proposition follows from
\[
S=\bigcup_{n=0}^\infty S_n.
\]
\end{proof}

To use Proposition~\ref{grading}, we modify \eqref{dotsdiagram} as:
\begin{equation}
\diagram{
 V_{(1)}\aul{r}{.3}{\pii{2}} & {X\x V\x \N}_{(2)}\aux[@<.4ex>]{d}{\varphi}\\
{X\x\N}_{(3)} \aux{r}{\id} & {X\x\N}_{(4)} \au[@<.6ex>]{u}{\pii{13}} \au{r}{\pi_1}
& X_{(5)}
}\label{dotsdiagram2}
\end{equation}
and define $\varphi: X\x V\x\N\to X\x\N$ by
$
\varphi((x,w,k))=(x+w,k+1)
$
as well as modifying \eqref{sec:dots} to
$
t(S_1)=\{v\},\;t(S_3)=\{(p,0)\}.
$
Then \eqref{sec:dotsS4} becomes
\begin{align}
s(S_4)&=s(S_3)\cup\varphi(s(S_2))\nonumber \\
&=\{(p,0)\}\cup\{(x+w,k+1)\| (x,k)\in s(S_4), w\in s(S_1)\} \label{sS4form}.
\end{align}
We define $g:s(S_4)\to\N$ by $g((x,k))=k$ and
$\eta:\2^{s(S_4)}\to \2^{s(S_4)}$ by
\[
\eta(A)=\{(x+w,k+1)\|(x,k)\in A, w\in s(S_1)\}.
\]
Then $g$ and $\eta$ clearly satisfy the condition of Proposition~\ref{grading}.
Thus it follows from \eqref{sS4form} and the proposition that
\begin{align*}
s(S_4)&= \bigcup_{n=0}^\infty \eta(\{(p,0)\}) \\
&=\{(p,0), (p+v,1), (p+2v,2), (p+3v,3), \cdots\}.
\end{align*}

Thus $\Gm(\S,\S',\M\,|\,t)$ contains only this cross section $s$ with $s(S_5)=D$,
and $D$ is represented by $(\S,\S',\M, \T, t, S_5)$.

If we set $t(S_1)=\{v,-v\}$, then
\[
s(S_4)=\{(p,0)\}\cup\{(x+v,k+1)\| (x,k)\in s(S_4)\}\cup\{(x-v,k+1)\| (x,k)\in s(S_4)\}
\]
and thus
\[
s(S_5)=\{\cdots, p-3v, p-2v, p-v, p, p+v, p+2v, p+3v, \cdots\}.
\]
Moreover, if we set $t(S_1)=\{v,-v, u, -u\}$, then in general we get a grid points
as shown in Figure~\ref{fig:grid}(a).

\begin{figure}[t]
\begin{center}
\includegraphics[scale=.88]{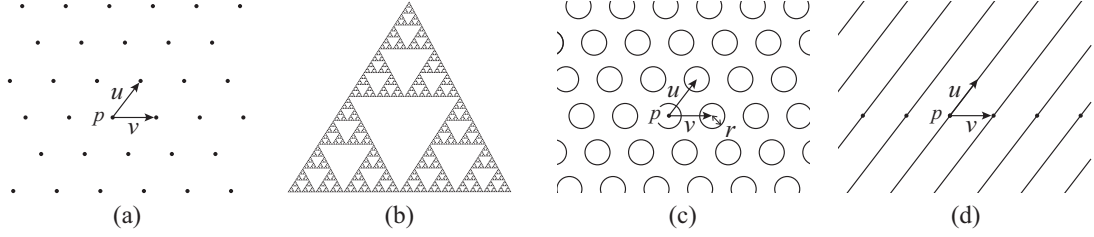}
\caption{(a) Recursively defined points. (b) Sierpinski triangle.
  (c) Hierarchically defined grid of circles. (d) Hierarchically defined lines.}
\label{fig:grid}
\end{center}
\end{figure}

\subsubsection{Sierpinski triangle}

Consider the diagram $(\S,\S',\M)$:
\[
\diagram[@C=4em]{
X_{(1)}\aul{r}{.4}{\pii{1}} & {X\x X}_{(2)} \au{r}{\pii{1}\o\sub} & {V\x\R}_{(3)} \au{d}{\sub\inv\o\mult}
 & {\R}_{(4)} \avl{l}{.4}{\pii{2}}\\
& X_{(5)} \au{r}{\pii{2}} \au{u}{\pii{2}} & {X\x X}_{(6)} \au{r}{\pi_1}  & X_{(7)}
}
\]
Here, let $\T=\{S_4, S_5\}$ and $t(S_4)=\left\{0.5\right\}, t(S_5)=\{p\}$.
Then for $s\in\Gm(\S,\S',\M\|t)$,
\begin{align*}
s(S_2)&=\{(x,p)\|x\in s(S_1)\}\\
s(S_3)&=\{(x-p,0.5)\|x\in s(S_1)\}\\
s(S_6)&=\{(y,p)\|0.5(x-p)=y-p, x\in s(S_1)\}\\
s(S_7)&=\{p+0.5(x-p)\| x\in s(S_1)\}.
\end{align*}
Let us define a map $d_p:X\to X$ for $p\in X$ by $d_p(x)=p+(x-p)/2$.
Then, the power map of $d_p$ is represented by $(\S,\S',\M, \T, t, S_1, S_7)$.

For three points $a, b$, and $c$ forming a triangle $T$ in $X$, let us define a map $\varphi_{abc}:2^X\to 2^X$ by
$\varphi_{abc}(A)=d_a(A)\cup d_b(A)\cup d_c(A)$ for $A\subset X$.
Obviously, $\varphi_{abc}$ can also be represented by a diagram.
Then, consider the diagram
\[
\diagram{
X_{(1)}\auxl{r}{.35}{\id\x 0} & {X\x\N}_{(2)}\arloopux{\varphi_{abc}'} \au{r}{\mathrm{cl}\o\pi_1} & {X}_{(3)}
}
\]
Here, $\varphi_{abc}'$ is $\varphi_{abc}$ that increments the second component and $\id\x 0$ maps
$A\subset X$ to $A\x\{0\}=\{(x,0)\|x\in A\}\subset X\x\N$.
Let $\T=\{S_1\}$ and let $t(S_1)=T$.
Then by Proposition~\ref{grading},
$s(S_2)= \bigcup_{n=0}^\infty \varphi_{abc}'(T\x\{0\})$.
Then $s\in\Gm(\S,\S',\M\|\T)$ gives as $s(S_3)$ the closure of $\bigcup_{n=0}^\infty \varphi_{abc}(T)$, which is
a fractal set known as the Sierpinski triangle.
See Figure~\ref{fig:grid}(b).

\subsection{Hierarchical Definition}
Combining \eqref{dia:circle} and \eqref{dotsdiagram2}, we consider $(\S,\S',\M)$:
\[
\diagram{
 V_{(1)}\aul{r}{.3}{\pii{2}} & {X\x V\x \N}_{(2)}\aux[@<.4ex>]{d}{\varphi} &
{\R}_{(3)} \au{r}{\len\inv} & {V}_{(4)} \au{d}{\sub\inv} \\
{X\x\N}_{(5)} \aux{r}{\id} & {X\x\N}_{(6)} \au[@<.6ex>]{u}{\pii{13}} \au{r}{\pi_1}
& X_{(7)} \aul{r}{.45}{\pii{2}} & {X\x X}_{(8)} \au{r}{\pi_1} & X_{(9)}
}
\]
The left-half comes from \eqref{dotsdiagram2} and produces the grid points in $S_7$,
which is the space of center points in the right-half \eqref{dia:circle}.
With $\T=\{S_1, S_3, S_5\}$ and
\[
t(S_1)=\{v,-v, u, -u\},\;\; t(S_3)=\{r\}, \;\;t(S_5)=\{(p,0)\},
\]
$(\S,\S',\M, \T, t, S_9)$ represents a grid of circles as shown in Figure~\ref{fig:grid}(c).

Similarly, with the diagram $(\S,\S',\M)$:
\[
\diagram{
V_{(1)} \aul{r}{.3}{\pii{2}} & {X\x V\x \N}_{(2)}\aux[@<.4ex>]{d}{\varphi} & V_{(3)} \aul{r}{.4}{\pii{1}} &
 {V\x\R}_{(4)} \aul{d}{.4}{\sub\inv\o\mult}\\
{X\x\N}_{(5)} \aux{r}{\id} & {X\x\N}_{(6)} \au[@<.6ex>]{u}{\pii{13}} \au{r}{\pi_1} &
X_{(7)} \aul{r}{.4}{\pii{2}}  & {X\x X}_{(8)}\au{r}{\pi_1} & X_{(9)}
}
\]
let $\T=\{S_1, S_3, S_5\}$ and
\[
t(S_1)=\{v,-v\},\;\; t(S_3)=\{u\}, \;\;t(S_5)=\{(p,0)\}.
\]
Then $(\S,\S',\M, \T, t, S_9)$ represents the lines that is parallel to $u$ and go through the points
$\{\cdots, p-3v, p-2v, p-v, p, p+v, p+2v, p+3v, \cdots\}$, as shown in Figure~\ref{fig:grid}(d).

\subsection{Two-part Coding}
\begin{figure}[t]
\begin{minipage}{.28\textwidth}
\[
\diagram{
{V}_{(1)}  \aul{r}{.4}{\pii{2}} & {X\x V}_{(2)} \au{d}{\add} \\
{X}_{(3)}  \aul{ur}{.4}{\pii{1}} & X_{(4)}
}
\label{dia:trans}
\]
\end{minipage}
\hfill
\begin{minipage}{.7\textwidth}
\centering
\includegraphics[scale=.99]{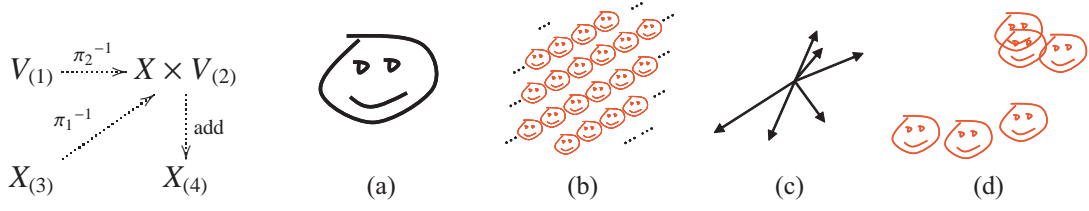}
\end{minipage}
\caption{A diagram to make repeated patterns with arbitrary patches. (a) A patch to be $B$: it can be any random data.
(b) Regularly repeated patches. (c) A set $A$ of random vectors. (d) Randomly repeated patches.
}\label{fig:smiley}
\end{figure}

It may seem that the representation can represent only very regular and simple objects such as geometric shapes like
 circles and lines,
except maybe by the trivial representation.
That is not the case.
It can represent the regular part of the object sparsely and the random part densely,
then mix them in various ways---intersection, union, hierarchically, and recursively---with
any computation, in fact, as we will see.
If there is regularity, it can be exploited.
For example, in the diagram on the left of Figure~\ref{fig:smiley},
we can take any pair of subsets $A\subset V$ and $B\subset X$
and consider a partial cross section $t$ such that $t(S_1)=A, t(S_3)=B$.
This gives us, for instance, repeated patterns with arbitrary patches by setting $A$
to be a regular grid by \eqref{dotsdiagram2} and $B$ to be the patch (Figure~\ref{fig:smiley}(b), with the patch (a)).
Using an arbitrary set $A$ as well (c), we can also have a randomly repeated patches (d).
This has still much more regularity than a completely random data.

Separating the information in an object into a regular part and the remaining random part is the basic idea
behind many formalisms such as Kolmogorov minimum sufficient statistic\cite{Li_Vitanyi_97} and minimum description
length (MDL)\cite{Rissanen78}.
Our representation can also utilize this basic idea, which is how simple objects can be represented more concisely
than random objects.

The representation of lines and circles by specifying points and vectors is simple;
With diagrams and cross sections, such simple representation and more complex ones involving any computation
can be uniformly embedded in the spaces that the objects reside.
We have shown a glimpse of such integration in the grid example and the Sierpinski triangle example.
In the next section, we discuss further the representation of computation by diagrams and their cross sections.

\section{Computation}\label{sec:computation}

With diagrams and their cross sections, we can represent computation.
In this section, we first examine a few examples of explicit representation of concrete computation.
We then show that any Turing machine can be represented by diagrams and cross sections.

\subsection{Examples}\label{compexample}

\subsubsection{Factorials}

Let $\suc:\N\to\N$ denote the successor function and $\mult:\N\x\N\to\N$ the multiplication.
Consider the following diagram $(\S,\S',\M)$:
\[
\diagram{
{\N\x \N}_{(1)} \aux{r}{\id} & {\N\x \N}_{(2)} \arlooprx{\varphi}
}
\]
where $\varphi=(\pi_1\x\mult)\circ((\suc\circ\pi_1)\x\pi_2)$ maps $(n,m)$ to $(n+1,m(n+1))$.
Suppose that $\T=\{S_1\}$ and that its cross section $t$ is defined by
\[
t(S_1)=\{(0,1)\}.
\]
Let $s$ be a cross section in $\Gm(\S,\S',\M\,|\,t)$.
Then we have
\begin{equation}
s(S_2)=s(S_1)\cup \varphi(s(S_2))=\{(0,1)\}\cup \varphi(s(S_2)).\label{s2s1s2}
\end{equation}

We define a grading function $g:\N\x\N\to\N$ by $g((n,m))=n$.
Then $g$ and $\varphi$ clearly satisfy the condition of Proposition~\ref{grading},
and it follows from \eqref{s2s1s2} that
\begin{equation}
s(S_2)= \bigcup_{n=0}^\infty \varphi(\{(0,1)\})
=\{(0,1), (1,1), (2,2), (3,6), \cdots, (n, n!), \cdots\}.\label{factorialsection}
\end{equation}

Therefore, $\Gm(\S,\S',\M\,|\,t)$ only contains the cross section $s$ defined by $s(S_1)=\{(0,1)\}$ and \eqref{factorialsection};
and $(\S,\S',\M, \T, t, S_2)$ represents the set of all pairs of natural number and its factorial.

We can extend the diagram to
\[
\diagram{
{\N\x \N}_{(1)} \aux{r}{\id} & {\N\x \N}_{(2)} \arlooprx{\varphi} \au{d}{\id}\\
{\N}_{(3)} \aul{r}{.4}{\pii{1}} & {\N\x \N}_{(4)} \au{r}{\pi_2} & {\N}_{(5)}
}
\]
Then, the (power map of the) factorial map is represented by $(\S,\S',\M, \T, t, S_3, S_5)$.

\subsubsection{Fibonacci Number}

Consider the following diagram $(\S,\S',\M)$:
\[
\diagram{
{\N^+\x \N^+}_{(1)} \aux{r}{\id} & {\N^+\x \N^+}_{(2)} \arloopdx{\varphi} \aul{r}{.6}{\pi_1} & {\N^+}_{(3)}
}
\]
where $\N^+$ is the set of positive integers, $\varphi=\pi_2\x\add$ with $\add :\N^+\x\N^+\to\N^+$ the addition.
Thus $\varphi((n,m))=(m,n+m)$.
Suppose that $\T=\{S_1\}$ and that its cross section $t$ is defined by $t(S_1)=\{(1,1)\}$.
Then if $s\in\Gm(\S,\S',\M\,|\,t)$ we have
\[
s(S_2)=\{(1,1)\}\cup\varphi(s(S_2)).
\]

If $(n,m)\in s(S_2)\setminus\{(1,1)\}$, it must be the case that $(n,m)\in\varphi(s(S_2))$
and $(m-n,n)\in s(S_2)$.
If this is not in $s(S_2)$, then $(2n-m,m-n)$ is in $s(S_2)$, and so on.
Since the sum of the two components decreases by this process, it cannot go on indefinitely
and has to stop by reaching $(1,1)$.
Thus
\[
s(S_2)= \bigcup_{n=0}^\infty \varphi(\{(1,1)\})=\{(1,1), (1,2), (2,3), (3,5), \cdots\}
\]
and $s(S_3)=\{1,1,2,3,5,8,\cdots\}$, which is the set of Fibonacci numbers.

\subsubsection{Mandelbrot Set}
Consider the diagram $(\S,\S',\M)$:
\[
\diagram[@C=4em]{
{\C\x\C\x\N}_{(1)} \arlooplx{\varphi} \au{r}{\pi_{12}} & {\C\x\C}_{(2)} \au{r}{\cmpl\o\pi_1} &  {\C_{(3)}}\\
{\C\x\C\x\N}_{(4)} \aux{u}{\id}  & {\C}_{(5)} \au{u}{\pii{2}}
}
\]
Here, $\C$ is the set of complex numbers.
The map $\varphi$ is defined by
\[
\varphi((c,z,k))=(c, c+z^2,k+1).
\]
Suppose that $\T=\{S_4,S_5\}$ and that its cross section $t$ is defined by
\begin{align*}
t(S_4)&=\{(c,0,0)\in\C\x\C\x\N\; | \;c\in\C\},\\
t(S_5)&=\{z\in\C\; | \;|z|>2\}
\end{align*}
Let $s$ be a cross section in $\Gm(\S,\S',\M\,|\,t)$.
Then we have
\[
s(S_1)=\varphi(s(S_1))\cup s(S_4)=\varphi(s(S_1))\cup \{(c,0,0)\in\C\x\C\x\N\; | \;c\in\C\}.
\]

Define $g:s(S_1)\to\N$ by $g((c,z,k))=k$.
Then by Proposition~\ref{grading},
\[
s(S_1)= \bigcup_{n=0}^\infty \varphi(s(S_4))
=\{(c,z^c_n,n)\in\C\x\C\x\N\| c\in\C, n\in\N\},
\]
where $z^c_n$ is defined by $z^c_0=0$ and $z^c_{n+1}=(z^c_n)^2+c$.

Thus $\Gm(\S,\S',\M\,|\,t)$ contains only one cross section $s$ given by:
\begin{align*}
s(S_1)&=\{(c,z^c_n,n)\in\C\x\C\x\N\| c\in\C, n\in\N\},\\
s(S_2)&=\{(c,z^c_n)\in\C\x\C\|c\in\C, n\in\N, |z^c_n|>2\},\\
s(S_3)&=\{c\in\C\|\nexistsa n\in\N, |z^c_n|>2\},\\
s(S_4)&=\{(c,0,0)\; | \;c\in\C\},\\
s(S_5)&=\{z\in\C\; | \;|z|>2\}.
\end{align*}

The Mandelbrot set is defined to be the set of complex numbers $c$ such that $z^c_n$ does not tend to infinity.
It is known that $z^c_n$ tends to infinity if and only if $|z^c_n|>2$ for some $n\in\N$.
Thus $s(S_3)$ is the Mandelbrot set.

\subsubsection{Sum}\label{subsec:sum}
Let $X$ be a finite set and let $(a_x)_{x\in X}$ be real numbers indexed by $X$.
Consider the following diagram $(\S,\S',\M)$:
\[
\diagram[@R=4em@C=3.7em]{
{X\x\R}_{(1)} \aux{r}{(s\o\pi_1)\x\pi_2} &
{\2^X\x\R}_{(2)} \ar@{.>}@(dl,ul)[d]_(.65){\pii{12}}
\ar@{.>}@<-.4ex>@(dr,ur)[d]^(.65){\pii{34}}
 \au{r}{\id} & {\2^X\x\R}_{(3)} \au{r}{\pi_2} & {\R}_{(4)}\\
{\2^X}_{(5)}  \au{ru}{\pii{1}\o\cmpl} \avl{r}{.35}{(\cap\o\pi_{13})\inv} &
{\2^X\x\R\x\2^X\x\R}_{(6)} \aux{u}{\varphi} & {\2^X}_{(7)} \au{u}{\pii{1}}
}
\]
where the maps are defined as:
\begin{align*}
s:X\ni x &\mapsto \{x\}\in\2^X\\
\cap:\2^X\x\2^X\ni (A,B) &\mapsto A\cap B\in\2^X\\
\varphi: \2^X\x\R\x\2^X\x\R\ni(A,x,B,y) &\mapsto (A\cup B, x + y)\in \2^X\x\R
\end{align*}

Suppose that $\T=\{S_5,S_7\}$ and that its cross section $t$ is defined by $t(S_5)=\{\emptyset\}$
 and $t(S_7)=\{X\}$.
Let $s$ be a cross section in $\Gm(\S,\S',\M\,|\,t)$ and assume
\[
s(S_1)=\{(x,a_x)\|x\in X\}.
\]
Then we have
\begin{align}
s(S_2)&=\left[\{(\{x\},a_x)\|x\in X\} \cup \varphi(s(S_6))\right]\cap \pii{1}(\cof s(S_5))\nonumber\\
&=\{(\{x\},a_x)\|x\in X\} \cup \{(A\cup B, a+b)\|(A,a,B,b)\in s(S_6), A\cup B\neq\emptyset\}\label{suminit}\\
s(S_6)&=\{(A,a,B,b)\|(A,a), (B,b)\in s(S_2), A\cap B=\emptyset\}
\end{align}
Let us denote $a_A = \sum_{x\in A}a_x$ for $A\subset X, A\neq\emptyset$.
\begin{prop}
\[
s(S_2)=\{(A, a_A)\|A\subset X, A\neq\emptyset\}.
\]
\end{prop}
\begin{proof}
Define $s_n=\{(A,a)\in s(S_2)\||A|=n\}$ and $t_n=\{(A, a_A)\|A\subset X, |A|=n\}$ for $n=1,2,\cdots,|X|$.
We use an induction on $n$ to prove $s_n=t_n$.
Suppose $s_1\ni (A,a)$.
If $(A,a)$ is not of the form $(\{x\},a_x)$ for some $x\in X$, it is in $\varphi(s(S_6))$ and there must be
$(B,b), (C,c)\in s(S_2)$ such that $A=B\cup C$ and $B\cap C=\emptyset$.
However, it is impossible since $|A|=1$ and there is no element in $s(S_2)$ of the form $(\emptyset,a)$.
Thus $s_1\subset t_n$.
Since $s_1\supset t_1$ by \eqref{suminit}, this proves the case $n=1$.
Now, suppose $n\geq 2$ and $s_n\ni (A,a)$.
Since $|A|\geq 2$, $(A,a)$ must be in $\varphi(s(S_6))$.
Then there exist $(B,b), (C,c)\in s(S_6)$ such that $A=B\cup C, B\cap C=\emptyset$, and $a = b+c$.
But by the induction hypothesis $b=a_B, c=a_C$ so $a = a_B+a_C=a_A$.
Thus $s_n\subset t_n$.
On the other hand, suppose $(A,a_A)\in t_n$ and $x\in A$.
Then $(\{x\},a_x)\in t_1 = s_1 \subset s(S_2)$ and
$(A\setminus\{x\}, a_{A\setminus\{x\}}) \in t_{n-1}=s_{n-1} \subset s(S_2)$
by the induction hypothesis.
Thus $(A,a_A)=\varphi((\{x\},a_x,A\setminus\{x\}, a_{A\setminus\{x\}}))\in s(S_2)$.
Thus $(A,a_A)\in s_n$ and $s_n\supset t_n$.
\end{proof}
It follows
\begin{align*}
s(S_3)&=\{(X,a_X)\}\\
s(S_4)&=\{a_X\}.
\end{align*}
Thus $(\S,\S',\M, \T, t, S_1, S_4)$ represents the map $\Sigma_X:\2^{X\x\R}\to\2^\R$ that satisfies
\[
\Sigma_X(\{(x,a_x)\|x\in X\})=\left\{\sum_{x\in X}a_x\right\}.
\]
For a general subset $A\subset X\x\R$, $\Sigma_X$ gives
\[
\Sigma_X(A)=\left\{\sum_{x\in X}f(x)\;\Biggr|\;f:X\to\R, \text{ s.t. }\forall x\in X, (x,f(x))\in A\right\}.
\]

\subsubsection{Markov Random Field}
A Markov random field (MRF) consists of an undirected graph $G = (V,E)$,
a finite set $L$ of labels, and an energy function $\E$ on the space $L^V$ of label assignments to vertices,
 or {\it configurations}.
The energy function must be of the form:
\[
\E(X) = \sum_{C \in \scr{C}} E_C(X),
\]
where $\scr{C}$ denotes the set of cliques in $G$ and $E_C$ a function on
$L^V$ with the property that $E_C(X)$ depends only on values of $X\in L^V$ on $C$.
An MRF is of the first-order if $E_C(X)$ is constant unless $C$ is a vertex or an edge, and thus
the energy can be written as:
\[
\E(X) = \sum_{v\in V} E_1(v,X(v)) + \sum_{(u,v)\in E} E_2(u,X(u),v,X(v))
\]
with $E_1:V\x L\to\R$ and $E_2:V\x L\x V\x L\to\R$.
We assume that $E_2(u,L,v,L')=0$ unless $(u,v)$ is an edge.
Solving an MRF involves finding the configuration $X$ with the minimum $\E(X)$.

Consider the following diagram $(\S,\S',\M)$:
\[
\diagram[@R=4em]{
& {V\x L + V\x L\x V\x L}_{(1)} \au{r}{\E'} &
  {{(V + V\x V)\x\R}_{(2)}} \aul{r}{.7}{\Sigma_{V + V\x V}} &
{\R}_{(3)} \av{d}{\ltsup}\\
{V\x L\x L}_{(4)} & {V\x L}_{(5)}
\ar@{.>}@(ul,ur)!UL+<.8333em,0em>;[l]!UR+<-.8333em,0em>_(.55){\pii{12}}
\ar@{.>}@(dl,dr)!DL+<.8333em,0em>;[l]!DR+<-.8333em,0em>_{\pii{13}}
\avl{l}{.45}{\cmpl\o\pi_{122}}
\au{r}{\pi_2}
\aul[@<.4ex>]{u}{.6}{(\id+\pi_{12})\inv}
\avl[@<-.4ex>]{u}{.6}{(\id+\pi_{34})\inv}
 & {V}_{(6)}  & {\R}_{(7)}
}
\]
Define the maps
\begin{align*}
E'_1=\pi_1\x E_1&: V\x L\to V\x\R\\
E'_2=\pi_1\x \pi_3\x E_2&: V\x L\x V\x L\to V\x V\x\R.
\end{align*}
Also, define the map
\[
\psi: V\x\R + V\x V\x \R\to(V + V\x V)\x\R
\]
by
\[
\psi=(\iota_1\circ\pi_1+\iota_2\circ\pi_{12})\x(\pi_2+\pi_3).
\]
Then
\[
\E'=\psi\circ(E'_1+E'_2):V\x L + V\x L\x V\x L\to (V + V\x V)\x\R
\]

The map $\Sigma_{V + V\x V}: (V + V\x V)\x\R\to\R$ is defined in \ref{subsec:sum}.
Also, define $\ltsup: 2^\R \to 2^\R$ by $\ltsup(A)=\{x\in\R\|\existsa y\in A, x\leq y\}$,
which can be represented by
\[
\diagram{
{\R} \aul{r}{.45}{\pii{2}}  & {\R\x\R} \aul{r}{.45}{\pi_1\x \leq} &
 {\R\x\2} \au{r}{\pi_1} & {\R}\\
& & {\2} \au{u}{\pii{2}}
}
\]
where the map $\leq:\R\x\R\to\2$ maps $(x,y)$ to $1$ if $x\leq y$ and $0$ otherwise,
and a partial cross section is defined by $s(\2)=\{1\}$.

Suppose that $\T=\{S_4, S_6\}$ and that its cross section $t$ is defined by $t(S_4)=\emptyset, t(S_6)=V$.
Let $s$ be a cross section in $\Gm(\S,\S',\M\,|\,t)$.
Then, from
\begin{align*}
s(S_4)&=\{(v,l_1,l_2)\| (v,l_1), (v,l_2)\in s(S_5)\}\cap \cof\{(v,l,l)\|(v,l)\in s(S_5)\}\\
&=\{(v,l_1,l_2)\| (v,l_1), (v,l_2)\in s(S_5), l_1\neq l_2\}\\
&=\emptyset
\end{align*}
it follows that there is at most one element $(v,l)\in s(S_5)$ for each $v\in V$.
From $s(S_6)=\pi_2(s(S_5))=V$, there must be one for each $v$, thus there is exactly one element
$(v,l)\in s(S_5)$ for each $v\in V$.
Thus, $s(S_5)$ defines a configuration $X_s \in L^V$ by defining $X_s(v)=l$ for each $(v,l)\in s(S_5)$.
Then,
\begin{align*}
s(S_1)=&\{(v,X_s(v))\|v \in V\}+\{(u,X_s(u),v,X_s(v)))\|u,v \in V\}\\
s(S_2)=&\{(v,E_1(v,X_s(v))))\|v\in V\}+\{((u,v),E_2(u,X_s(u),v,X_s(v))))\|u,v \in V\}\\
s(S_3)=&\{\E(X_s)\}\\
s(S_7)=&\{x\in\R\|x\leq\E(X_s)\}
\end{align*}
Thus by minimizing on $S_7$, $\min_{S_7}\hspace{-.1667em}\Gm(\S,\S',\M\,|\,t)$ gives the cross sections in $s(S_5)$
that give the configurations with the minimum energy in $s(S_3)$.
Obviously, higher order MRFs can be treated in the same manner.

\subsubsection{Finite Automaton}
Let $(Q, \Sigma, \delta, q_0, F)$ be a finite automaton, where $Q$ and $\Sigma$ are the finite
sets of states and symbols, respectively, and $\delta : \Sigma\x Q \to Q$ is the transition function,
$q_0 \in Q$ is the initial state, and $F \subset Q$ is the accepting subset of $Q$.

Consider the diagram $(\S,\S',\M)$:
\[
\diagram[@C=4em]{
{\Sigma^{*}}_{(1)} \au{rd}{\pii{1}} & {\Sigma^{*}\x Q\x\N}_{(2)}
\ar@{.)}@(r,rd)!R+<0em,-.25em>;!R+<0em,-.6667em>^{\delta^{*}} \au{r}{\pi_{12}}
&  {\Sigma^{*}\x Q}_{(3)}  \au{rd}{\pi_2}  & Q_{(4)} \avl{l}{.3}{\pii{2}} \\
{Q\x\N}_{(5)} \aul{r}{.4}{\pii{23}} & {\Sigma^{*}\x Q\x\N}_{(6)} \avxl{u}{.4}{\id} &
{\Sigma^{*}}_{(7)} \au{u}{\pii{1}} & {Q}_{(8)}
}
\]
Here, $\Sigma^{*}$ is the set of strings on $\Sigma$ and
\[
\delta^{*}:\Sigma^{*}\x Q\x\N \ni(\tau,q,k)\mapsto
\begin{cases}
(\tau,q,k+1) & \text{ if }\tau=\epsilon\\
(\cdr(\tau), \delta(\car(\tau), q),k+1) & \text{ otherwise}
\end{cases}
\]
where $\car(\tau)\in \Sigma$ is the first symbol
in the string $\tau$ and $\cdr(\tau)\in\Sigma^{*}$ is the rest of the string.

Suppose that $\T=\{S_1,S_4,S_5,S_7\}$ and that its cross section $t$ is defined by
\[
t(S_1)=\{\sigma\},\;\; t(S_4)=F,\;\; t(S_5)=\{(q_0,0)\},\;\; t(S_7)=\{\epsilon\},
\]
where $\sigma\in\Sigma^{*}$ is a string and $\epsilon$ denotes the empty string.

Let $s$ be a cross section in $\Gm(\S,\S',\M\,|\,t)$.
Then
\begin{align*}
s(S_6)&=\{(\sigma, q_0, 0)\},\\
s(S_2)&=s(S_6) \cup \delta^{*}(s(S_2)).
\end{align*}
Define $g:s(S_2)\to\N$ by $g((\tau,q,k))=k$.
Then by Proposition~\ref{grading},
\[
s(S_2)=\bigcup_{n=0}^\infty \delta^{*}(s(S_6))=\{(\delta^{*})^n((\sigma, q_0, 0))\|n=0,1,2,\cdots\}.
\]
Thus $s(S_2)$ is the whole execution history of the automaton, beginning with $(\sigma, q_0, 0)$.
Therefore, from
\begin{align*}
s(S_3)&=\{(\epsilon, q)\| (\epsilon, q, k)\in s(S_2), q\in F\},\\
s(S_8)&=\{q\| q\in F, \existsa (\epsilon, q) \in s(S_3)\},
\end{align*}
$s(S_8)$ is nonempty if and only if the automaton accepts $\sigma$.

\subsection{Turing completeness}

\subsubsection{Turing Machine}
Let $M=(Q, \Sigma, \Theta, \text{\textvisiblespace}, \delta, q_0, \qA, \qR)$ be a Turing machine, where $Q$ is
the finite set of states, $\Sigma$ is the input alphabet, $\Theta$ is the finite set of internal
symbols such that $\Sigma\subset\Theta$,
$\bs$ is the blank symbol in $\Theta\setminus\Sigma$,
$\delta : \Theta\x Q \to \Theta\x Q\x\{-1,1\}$ is the transition function, and
$q_0, \qA, \qR \in Q$ are the initial, accept, and reject states, respectively.

Consider the following diagram $(\S,\S',\M)$:
\begin{equation}
\diagram[@R=3.5em@C=4em]{
{\Sigma^{*}}_{(1)} \aul{rd}{.4}{\pii{1}} & {\Theta^{*}\x Q\x \N\x \N}_{(2)}
\ar@{.)}@(rd,d)!RD+<-1.6667em,0em>;!RD+<-2.5em,0em>^{\delta^{*}} \aul{r}{.6}{\pi_2} & {Q}_{(3)} \\
{Q\x \N\x \N}_{(4)} \aul{r}{.4}{\pii{234}} &
{\Theta^{*}\x Q\x \N\x\N}_{(5)} \avx{u}{\id} & Q_{(6)} \au{u}{\pii{2}}
}\label{TuringMachineDiagram}
\end{equation}
Among the components of $\Theta^{*}\x Q\x\N\x\N\ni(\sigma, q, n, k)$, the first three represent
a configuration of the machine, with $\sigma$ the content of the tape, $q$ the state of the machine,
and $n$ the position of the head.
The last component $k$ is a step number.
The map $\delta^{*}:\Theta^{*}\x Q\x\N\x\N\to\Theta^{*}\x Q\x\N\x\N$ is defined so that
it updates the configuration of the machine by one step, as follows.

For a string $\sigma$, $\sigma[n]$ denotes the $n$'th symbol in the string.
We use the zero-based index, so $\sigma[0]$ denotes the first symbol of $\sigma$.

First, we define
\begin{equation}
\eta:\Theta^{*}\x Q\x\N\x\N\x(\Theta\x Q\x\{-1,1\})\to\Theta^{*}\x Q\x\N\x\N.\label{TuringEta}
\end{equation}
Suppose $(\sigma, q, n, k, (x, q', v))\in\Theta^{*}\x Q\x\N\x\N\x(\Theta\x Q\x\{-1,1\})$.
If $q=\qA$ or $q=\qR$, $\eta(\sigma,q,n,k, (x, q', v))=(\sigma,q,n,k+1)$.
Otherwise, $\eta(\sigma,q,n,k, (x, q', v))=(\sigma',q',n',k+1)$, where $n'=n+v$ unless $n+v<0$, in which case $n'=0$.
As for $\sigma'$, it is the string made from $\sigma$ by (i) replacing $\sigma[n]$ by $x$ if $0\leq n \leq |\sigma|-1$,
or (ii) appending $x$ if $n=|\sigma|$, where $|\sigma|$ denotes the length of the string $\sigma$.

Next, define $\theta:\Theta^{*}\x Q\x\N \to \Theta\x Q\x\{-1,1\}$ by
\[
\theta(\sigma, q, n)=
\begin{cases}
\delta(\sigma[n], q) & \text{if }0\leq n\leq |\sigma|-1\\
\delta(\text{\textvisiblespace}, q) & \text{otherwise}
\end{cases}
\]
for $(\sigma, q, n)\in\Theta^{*}\x Q\x\N$.

Then $\delta^{*}$ is defined by
\[
\delta^{*}(\sigma, q, n, k) = \eta(\sigma, q, n, k, \theta(\sigma, q, n)).
\]

Suppose that $\T=\{S_1,S_4,S_6\}$ and that its cross section $t$ is defined by
\[
t(S_1)=\{\sigma\},\;\; t(S_4)=\{(q_0,0,0)\},\;\; t(S_6)=\{\qA, \qR\},
\]
where $\sigma\in\Sigma^{*}$.

Let $s$ be a cross section in $\Gm(\S,\S',\M\,|\,t)$.
Then
\begin{align*}
s(S_5)&=\{(\sigma, q_0,0,0)\},\\
s(S_2)&=\{(\sigma, q_0, 0, 0)\} \cup \delta^{*}(s(S_2)).
\end{align*}
Define $g:s(S_2)\to\N$ by $g((\tau,q,n,k))=k$.
Then by Proposition~\ref{grading},
\[
s(S_2)=\bigcup_{n=0}^\infty {\delta^{*}}^n(s(S_5))=\{(\delta^{*})^n((\sigma, q_0, 0, 0))\|n=0,1,2,\cdots\}.
\]
Thus $s(S_2)$ is the whole execution history of the Turing machine $M$, beginning with $(\sigma, q_0, 0, 0)$.
Therefore, if $M$ accepts the string $\sigma$, $s(S_3)=\{\qA\}$;
if it rejects $\sigma$, $s(S_3)=\{\qR\}$; and if it does not terminate,
$s(S_3)=\emptyset$.

\subsubsection{Nondeterministic Turing Machine}
The case of nondeterministic Turing machine is similar to the deterministic case.
Let $N=(Q, \Sigma, \Theta, \text{\textvisiblespace}, \delta', q_0, \qA, \qR)$ be a
nondeterministic Turing machine, where $Q$, $\Sigma$, $\Theta$, \textvisiblespace, $q_0$, $\qA$, and
$\qR$ are as before, and
$\delta' : \Theta\x Q \to \2^{\Theta\x Q\x\{-1,1\}}$ is the nondeterministic transition function
that gives the set of possible transitions.

The same diagram \eqref{TuringMachineDiagram} as the deterministic case will do with just a modification of the definition
of $\delta^{*}:\2^{\Theta^{*}\x Q\x\N\x\N}\to\2^{\Theta^{*}\x Q\x\N\x\N}$ as follows:
\[
\delta^{*}(A) = \{\eta(\sigma, q, n, k, \tau)\;|\:(\sigma, q, n, k)\in A, \tau\in\theta'(\sigma, q, n)\},
\]
where $\theta':\Theta^{*}\x Q\x\N \to \2^{\Theta\x Q\x\{-1,1\}}$ is defined by
\[
\theta'(\sigma, q, n)=
\begin{cases}
\delta'(\sigma[n], q) & \text{if }0\leq n\leq |\sigma|-1\\
\delta'(\text{\textvisiblespace}, q) & \text{otherwise}
\end{cases}
\]
and $\eta$ is as before \eqref{TuringEta}.

The cross section has $s(S_9)\ni \qA$ if and only if $N$ accepts the string $\sigma$.

\section{Information Measure of Structure}\label{sec:measure}

We define an information measure of general subsets.
The measure takes into account the structure of the set characterized by maps,
which we call the {\it structure maps}.
Thus, it is a measure of information relative to the structure maps.
Essentially, the measure is the size of the smallest diagram that represents the given subset,
in which the maps in the diagram can be written in terms of the structure maps.

\subsection{Maps Generating Diagrams}

\begin{defn}
Let $\M$ be a set of maps.
The set $\Mg$ of maps {\it generated by} $\M$ is defined as follows:
\begin{enumerate}[a)]

\item If $f$ is in $\M$, or it is $\id$, $\omega$, or a projection map $\pi_i$, then $f$ is in $\Mg$.

\item If maps $f:X\to Y, g:Y\to Z$ are in $\Mg$, then $g\circ f:X\to Z$ is also in $\Mg$.

\item If maps $f_i:X\to Y_i\;(i=1,\cdots, n)$ are in $\Mg$, then the product map
$f_1\x\cdots\x f_n:X\to Y_1\x\cdots\x Y_n$ is also in $\Mg$.
\end{enumerate}
\end{defn}

Note that for any $X$ and $x\in X$, $x:\1\to X$ is a map, which can be contained in the set $\M$ of maps in
this context.
When $\M$ contains all maps $x:\1\to X$, we write $X\subset\M$.

\begin{defn}
Let $\M$ be a set of maps.
For a map $f$ in $\Mg$, its {\it size} $|f|_\M$ relative to $\M$ is defined as follows:
\begin{enumerate}[a)]
\item
If $f$ is in $\M$, or it is $\id$, $\omega$, or a projection map $\pi_i$, then $|f|_\M=1$.

\item Otherwise, $f$ is either a concatenation or a product of some maps in $\Mg$.
In this case, $|f|_\M$ is $1$ plus the minimum of the sum of the size of component maps
relative to $\M$, among all possible combination, i.e.,
\begin{enumerate}[i)]
\item if $f\in\Mg$ is a concatenation,
\[
|f|_\M=1+\min_{\substack{g, h\in\Mg\\ f=g\o h}} (|g|_\M+|h|_\M),
\]
\item if it is a product,
\[
|f|_\M=1+\min_{\substack{f_1, \cdots, f_n\in\Mg\\ f=f_1\x\cdots\x f_n}} (|f_1|_\M+\cdots+|f_n|_\M),
\]
\end{enumerate}
\end{enumerate}
\end{defn}

\begin{defn}
Let $\MS$ be a set of maps.
A diagram $(\S,\S',\M)$ is said to be {\it generated by} $\MS$ if each map in $\M$ is either
\begin{enumerate}[a)]
\item the power map $f:\2^X\to\2^Y$ of a map $f:X\to Y$ in $\MSg$ or

\item the inverse map $f\inv:\2^Y\to\2^X$ of a map $f:X\to Y$ in $\MSg$.
\end{enumerate}
\end{defn}

\begin{defn}
Let $\MS$ be a set of maps.
A diagram $(\S,\S',\M)$ is said to be {\it generated by} $\MS$ {\it with complement} if each map in $\M$ is either
\begin{enumerate}[a)]
\item the power map $f:\2^X\to\2^Y$ of a map $f:X\to Y$ in $\MSg$,

\item the inverse map $f\inv:\2^Y\to\2^X$ of a map $f:X\to Y$ in $\MSg$, or

\item a complement map $\cmpl:\2^X\to \2^X$ for a set $X$ in $\S$.
\end{enumerate}
\end{defn}

\begin{defn}
Let $\MS$ be a set of maps and $(\S,\S',\M)$ a diagram generated by $\MS$ with or without complement.
For a map $\varphi$ in $\M$, its size $|\varphi|_{\MS}$ relative to $\MS$ is defined as follows:
\begin{enumerate}[a)]
\item if $\varphi=\cmpl$, $|\varphi|_{\MS}=1$.
\item If $\varphi=f$ or $\varphi=f\inv$ with $f\in\MSg$, then $|\varphi|_{\MS}=|f|_{\MS}$.
\end{enumerate}
\end{defn}

In the case of the circle example \eqref{dia:circle}, the diagram is generated by $\MS=\{\len, \sub\}$
and each map in the diagram has the size 1 relative to $\MS$.
The diagram \eqref{linediagram} for line  is generated by $\{\sub, \mult\}$.

\subsection{Definition}
We would like to define a measure of structural information of a subset $A\subset X$ relative to a
fixed set $\MS$ of structure maps by considering all data $(\S,\S',\M,\T,t,X)$ that represents $A$
and such that the diagram $(\S,\S',\M)$ is generated by $\MS$, and then taking the minimum of the
total size of the maps appearing in such diagrams.
However, such a measure is not very useful without any restriction,
since for any $A\subset X$ and any $\MS$, such minimum is zero because of the trivial representation
in which $\S=\T=\{X\}, \M=\emptyset,$ and $t$ is defined by $t(X)=A$.

To remedy this problem, we balance the data supplied by the partial cross section with their
Shannon information\cite{Shannon48}.
\begin{defn}\label{withshannon}
Let $X$ be a set, $A$ its subset, and $\MS$ a set of maps.
Also, let $\P$ be a set of probability spaces, each space $Y\in\P$ with a probability measure $P_Y$.
We define the structural information $I(A\,|\,\MS,\P)$ of $A$ relative to $\MS$ and $\P$ as the minimum of
\[
\sum_{\varphi\in\M}|\varphi|_{\MS}-\sum_{T\in\T} \log_2 P_T(t(T))
\]
among all data $(\S,\S',\M,\T,t,X)$ representing $A$ such that $(\S,\S',\M)$ is a diagram generated by $\MS$,
each set in $\T$ is in $\P$, and $t(T)$ is a measurable subset of each $T\in\T$.
If there does not exist such a diagram with which the sum is finite, we define $I(A\,|\,\MS,\P)=\infty$.
We define the measure $I^\mathrm{c}(A\,|\,\MS,\P)$ {\it with complement} similarly except that we
allow diagrams generated by $\MS$ with complement.
\end{defn}

Imagine that $X$ is itself a probability space and $X\in\P$.
Then the existence of trivial representation gives the upper bound:
\[
I(A\,|\,\MS,\P)\leq -\log_2 P_X(A).
\]

In general, the sets in $\T$ can be thought of as parameter spaces, and the Definition~\ref{withshannon}
allows for the treatment of information in an ensemble of objects, just like Shannon information
but taking the structure into account.
However, in this paper, we focus on the information in an individual object.
Thus, for the rest of this paper, we only consider the special case of $I(A\,|\,\MS,\P)$ with
$\P=\{\1\}$, making $\1$ a probability space with the probability measure $P_\1$ defined by
$P_\1(\emptyset)=0, P_\1(\1)=1$.
That is, we require $\T=\{\1\}$ and $t(\1)=\1$.

\begin{defn}\label{infomeasuredef}
The structural information measures $I(A\,|\,\MS)$ and $I^\mathrm{c}(A\,|\,\MS)$ are defined as
$I(A\,|\,\MS,\P)$ and $I^\mathrm{c}(A\,|\,\MS,\P)$ with $\P=\{\1\}$, respectively.
\end{defn}

For instance, the diagram for circle \eqref{dia:circle} can be modified to:
\[
\diagram{
{\1}_{(1)} \au{r}{r} & {\R}_{(2)} \au{r}{\len\inv} & {V}_{(3)} \aul{d}{.45}{\sub\inv} \\
{\1}_{(4)} \au{r}{p} & {X}_{(5)} \aul{r}{.4}{\pii{2}} & {X\x X}_{(6)} \aul{r}{.55}{\pi_1} & X_{(7)}
}
\]
The constants $r\in\R$ and $p\in X$ are identified with the maps $r:\1\to \R$ and $p:\1\to X$
with values $r(0)=r\in\R$ and $p(0)=p\in X$ (see \ref{notations}.)
Remember that we omit the $\1$ in diagrams.
Thus the diagram for line \eqref{linediagram} can be modified thus:
\[
\diagram{
V_{(1)} \aul{r}{.4}{\pii{1}} \aiul{v}
& {V\x\R}_{(2)} \aul{d}{.45}{\sub\inv\o\mult}\\
X_{(3)} \aul{r}{.4}{\pii{2}} \aiul{p} & {X\x X}_{(4)}\aul{r}{.55}{\pi_1} & X_{(5)}
}
\]
The diagram is generated by $\MS=\{p, v, \sub, \mult\}$.
The map $\sub\inv\circ\mult$ is actually not in $\MSg$.
However, in general, if we have $\varphi:\2^S\to\2^T$ and $\psi:\2^T\to\2^U$, we can always
think of $\raisebox{.291667em}{\xymatrix{S \au{r}{\psi\o\varphi} & U}}$ as an abbreviation of
$\raisebox{.291667em}{\xymatrix{S \au{r}{\varphi} & T \au{r}{\psi} & U}}$.

This way, visual patterns such as shown in Figure~\ref{fig:example1} and \ref{fig:grid},
as well as explained in Section~\ref{sec:geopat}, can be shown to have finite information.
The circle $C$ has the upper bound of $I(C\,|\{r,p,\len,\sub\})\leq 6$,
while the line $L$ has $I(L\,|\{p,v,\sub,\mult\})\leq 7$.
When we talk about Euclidean space, we should include all maps that characterize the space $X$ and
related spaces of $\R, V$, which would include $\len, \sub, \mult$ and also any constants.
Thus, the proper set of maps to estimate the structural information of subsets of Euclidean space
would be something like $\M_\mathrm{E}=\{\len, \sub, \mult\}\cup X\cup V\cup\R$.
The above estimates for the circle and line cases are unchanged: $I(C\,|\M_\mathrm{E})\leq 6,
I(L\,|\M_\mathrm{E})\leq 7$.

\section{Relation to Kolmogorov Complexity}\label{sec:RelKolm}

In this section, we show that the information measure defined in the previous section
generalizes Kolmogorov complexity.
Here, we are only concerned with binary strings;
an extension to other alphabet should be straightforward.
For a Turing machine $M$ with binary alphabet, let the partial function defined by $M$
be denoted by the same letter.
Also, let $|\sigma|$ denote the length of the string $\sigma$.

\begin{defn}
The Kolmogorov complexity $K_U(\sigma)$ of a binary string $\sigma\in\2^*$ with respect to
 a universal Turing machine $U$ is defined as
\[
	K_U(\sigma)=\min_{p\in\2^*, U(p)=\sigma} |p|.
\]
\end{defn}

\subsection{Emulating Turing Machines}

It turns out that any Turing machine can be emulated by a diagram generated by
$\MN=\{0,\suc\}$, where $\suc:\N\to\N$ is the successor function and $0:\1\to\N$,
if we represent a binary string $\sigma\in\2^*$ by the subset
\[
\bsigma=\{(i,\sigma[i])\| i=0,1,\cdots,|\sigma|-1\}\cup\{(|\sigma|,0),(|\sigma|,1)\} \subset \N\x\N.
\]

\begin{thm}\label{TMEmulate}
Let $M$ be a Turing machine with binary alphabet.
Then there exists a finite diagram $(\S,\S',\M)$ generated by $\MN=\{0,\suc\}$
and two sets $S=\N\x\N$ and $T=\N\x\N$ in $\S$ such that
$(\S,\S',\M, \{\1\}, s_\1, S, T)$ represents the map $\bar{M}:\2^{\N\x\N}\to\2^{\N\x\N}$ that
maps $\bsigma$ to $\bar{\tau}$ if $M$ accepts $\sigma$ leaving a string
$\tau$ on its tape, or to $\emptyset$ if $M$ rejects $\sigma$ or does not terminate
with input $\sigma$.
\end{thm}
\begin{proof}
Let $M=(Q, \2, \Theta, \bs, \delta, q_0, \qA, \qR)$, where $Q$ is
the finite set of states, $\Theta$ is the finite set of internal symbols such that $\2\subset\Theta$,
$\bs\in\Theta$ is the blank symbol,
$\delta:\Theta\x Q \to \Theta\x Q\x\2$ is the transition function, and
$q_0, \qA, \qR \in Q$ are the initial, accept, and reject states, respectively.
We define $\Qt=Q+\{\qe\}$ and $\Tt=\Theta+\{\terminal\}$.
Consider the following diagram $(\S,\S',\M)$:
\begin{equation}
\diagram{
{\N\x\N}_{(1)} \auxl{r}{.35}{\init} &
{\N\x\Tt\x\Qt\x\N}_{(2)}
\ar@{.)}@(ru,u)!RU+<-1.6667em,0em>;!RU+<-2.5em,0em>_(.2){\delta^{*}} \au{r}{\id} &
{\N\x\Tt\x\Qt\x\N}_{(3)} \aul{r}{.6}{\theta\o\pi_{12}} \aiur{\pii{3}\o\qA}
& {\N\x\N}_{(4)}\\
}\label{TMEmuDiagram}
\end{equation}
Here, $S_2$ is meant to represent the whole execution history of the Turing machine $M$.

Let us assume that $\sigma$ is a binary string and $s$ is a cross section in $\Gm(\S,\S',\M\|s_\1)$ with $s(S_1)=\bsigma$.
For $k\in\N$, let us call the set $R_k=\{(i,x,q,k)\in s(S_2)\|(i,x,q)\in\N\x\Tt\x\Qt\}$
the $k$'th row.
We call the $R_k$ {\it proper} if the following condition holds:
\begin{align*}
R_k&=R^1_k+R^2_k;\\
R^1_k&=\{(i,x_{ik},\qe,k)\|i=0,\cdots,l_k\}, x_{l_kk}=\terminal, x_{ik}\neq\terminal\text{ for }i<l_k;\\
R^2_k&=\{(i,x_{ik},q_k,k)\|i=m_k,\cdots,n_k\}, q_k\neq\qe, 0\leq m_k\leq n_k\leq l_k;
 \text{ if }q_k\neq\qA\text{ then }m_k=n_k.
\end{align*}
In a proper row, $x_{ik}$ represents the $i$'th symbol on the tape at $k$'th step in execution of
$M$, except for the last entry in the row, which is $\terminal$.
When $q_k\neq\qA$, the single element $(m_k,x_{m_kk},q_k,k)$ in $R^2_k$ indicates that $M$ is in state $q_k$
and its head is at position $m_k$.
When $q_k=q_A$ it means that the machine is in the accept state.

The map $\init$ sends the subset $\bsigma$ to
\begin{equation}\label{initmap}
\init(\bsigma)=\{(i,\sigma[i],\qe,0)\| i=0,\cdots,|\sigma|-1\}\cup\{(|\sigma|,\terminal,\qe,0)\}\cup\{(0,\sigma[0],q_0,0)\},
\end{equation}
setting up the $0$'th row properly.

The map $\delta^*:\2^{\N\x\Tt\x\Qt\x\N}\to\2^{\N\x\Tt\x\Qt\x\N}$
maps $R_k$ to $R_{k+1}$ with the following properties:
\begin{enumerate}
\item If $R_k$ is proper and $q_k\neq\qA$, $\delta^*$ updates the
configuration by emulating $M$, i.e., $R^1_k$ is copied to $R^1_{k+1}$ except that $x_{m_kk}$ is
mapped to $x_{m_kk+1}$ according to $M$'s tape rewriting,
and that the single element $(m_k, x_{m_kk}, q_k,k)$ of $R^2_k$ is mapped to $(m_{k+1}, x_{m_{k+1}k},q_{k+1},k+1)$
in $R^2_{k+1}$, where $m_{k+1}$ is the position of the head at step $k+1$.

\item If $k$'th row is proper and $q_k=\qA$, $R^1_k$ is copied to $R^1_{k+1}$ and $R^2_k$ is expanded.
That is, if
\[
R^2_k=\{(i,x_{ik},\qA,k)\|i=m_k,\cdots,n_k\},
\] then
\[
R^2_{k+1}=\{(i,x_{ik},\qA,k)\|i=\max(0,m_k-1),\cdots,\min(l_k, n_k+1)\}.
\]
\end{enumerate}
In Lemma \ref{deltastar}, we show that such a $\delta^*$ can be represented using a subdiagram.
If we define a grading map $g:\N\x\Tt\x\Qt\x\N\to\N$ by $g((i,x,q,k))=k$,
it is clear from the proof of the lemma that $\delta^*$ satisfies the requirement of
Proposition~\ref{grading}, which is
\[
\delta^*(S_i)\subset S_{i+1},\hspace{2em} \delta^*(S)=\bigcup_{n=0}^\infty \delta^*(S_n),
\]
where $S=\N\x\Tt\x\Qt\x\N$ and $S_i=g\inv(i)$.
Thus from the proposition we have
\[
s(S_2)=\bigcup_{k=0}^\infty {\delta^{*}}^k(\init(s(S_1)))
\]
We then have the following:
\[
s(S_3)=\{(i,x,q,k)\in s(S_2)\|q=\qA\}.
\]
If $M$ reaches the accept state at $k$'th step, $R^2_k$ would have one element with the accept state $\qA$.
The rows after that would have increasingly more $\qA$ without changing the symbols,
until $R^2_{k'}=\{(i,x_{ik},\qA,k')\|i=0,\cdots,l_k\}$.
The map $\theta:\2^{\N\x\Tt}\to\2^{\N\x\N}$ just fixes the termination of the string,
removing the blank symbols that might be left at the end of the string and properly terminating:
\begin{equation}
\begin{aligned}
\theta(A)=&\{(i,x)\|(i,x),(i+1,y)\in A,x,y\in\2\}\cup\{(0,0),(0,1)\|(0,y)\in A, y\in B\}\cup\\
&\{(i,x),(i+1,0),(i+1,1)\|(i,x),(i+1,y)\in A,x\in\2, y\in B\},
\end{aligned}\label{termmap}
\end{equation}
where $B=\{\bs,\terminal\}$.
The second part of the RHS is to take care of the case when the string is empty.
Thus, $s(S_4)$ contains the string that remains when $M$ reaches the accept state if it does,
and is empty if it does not.

Without loss of generality, we can assume that $Q=\n, \qe=n$ and $\Theta=\m, \terminal=m$.
Remember, in our notation, the finite subset $\{0,1,\cdots,n-1\}$ of $\N$ is denoted by $\n$.
Thus, constant maps such as $\qe, q_0$ are natural number constants.
Any natural number constant map $k$ can be made from $0$ and $\suc$ as $\suc^k\circ 0$.

It remains to prove that the maps  $\init$, $\theta$, and $\delta^*$ can be represented by a diagram
generated by $\MN$, which is done in the following lemmas.
\end{proof}

\begin{lem}
The map $\init:\2^{\N\x\N}\to\2^{\N\x\Tt\x\Qt\x\N}$ that satisfies
\eqref{initmap} can be represented by a diagram generated by $\MN$.
\end{lem}
\begin{proof}
Consider the diagram
\[
\diagram[@C=5.5em]{
{\N\x\N}_{(1)} \au[@<0em>]{d}{(\id\x 1)\inv} \av[@<-0.333em>]{d}{(\id\x 0)\inv}
\ar@{.)}@(ur,ul)!UR+<-.5em,0em>;[rr]!UL+<0em,0em>_{\pi_1\times\qe}
\av[@<-.1667em>]{r}{\id} \au[@<.1667em>]{r}{\pii{1}\o\suc\inv\o\pi_1} &
{\N\x\N}_{(2)} \aux{d}{\id} &
{\N\x\Qt}_{(3)} \ajurx{0\x q_0} \au{d}{\pii{13}}\\
{\N}_{(4)} \aux{r}{\id\x\terminal}&
{\N\x\Tt}_{(5)} \aul{r}{.4}{\pii{12}} &
{\N\x\Tt\x\Qt\x\N}_{(6)}
\ar@{.>}@!+<-2.5em,1.8333em>;[0,0]+<-2em,.83333em>_(.0){\pii{4}\o 0}
}
\]
Let $\sigma$ be a binary string and $s$ a cross section of the diagram such that
$s(S_1)=\bsigma=\{(i,\sigma[i])\| i=0,1,\cdots,|\sigma|-1\}\cup\{(|\sigma|,0), (|\sigma|,1)\}$.
Then
\begin{align*}
s(S_2)&=\{(i,x)\|(i+1,x')\in s(S_1), x'\in\2, (i,x)\in s(S_1)\}=\{(i,\sigma[i])\| i=0,\cdots,|\sigma|-1\},\\
s(S_3)&=\{(i,\qe)\|(i,x)\in s(S_1), x\in\2\}\cup\{(0,q_0)\}=\{(i,\qe)\|i=0,\cdots,|\sigma|\}\cup\{(0,q_0)\},\\
s(S_4)&=\{i\in\N\|(i,0),(i,1)\in s(S_1)\}=\{|\sigma|\},\\
s(S_5)&=\{(i,\sigma[i])\|i=0,1,\cdots,|\sigma|-1\}\cup\{(|\sigma|,\terminal)\},\\
s(S_6)&=\{(i,\sigma[i],\qe,0)\| i=0,\cdots,|\sigma|-1\}\cup\{(|\sigma|,\terminal,\qe,0)\}\cup\{(0,\sigma[0],q_0,0)\}.
\end{align*}
Thus, $s(S_6)$ is the result of applying $\init$ to $s(S_1)$.
\end{proof}

\begin{lem}
The map $\theta:\2^{\N\x\Tt}\to\2^{\N\x\N}$ that satisfies
\eqref{termmap} can be represented by a diagram generated by $\MN$.
\end{lem}
\begin{proof}
Consider the diagram
\[
\diagram[@R=3em@C=4em]{
{\N\x\Tt}_{(1)} \auxl{d}{.3}{(\suc\o\pi_1)\x\pi_2} \aul{r}{.45}{\pii{13}}&
{\N\x\Tt\x\Tt}_{(2)} \au{r}{\id} \au{dr}{\id} &
{\N\x\Tt\x\Tt}_{(3)}
\aux{rd}{\alpha}
\ar@{.)}@(dr,l)!D+<0em,0em>;[dr]!L+<0em,.5em>_(.2){\pi_1\x 0}
\ar@{.)}@(r,ul)!D+<2.8333em,.333em>;[dr]!L+<.8333em,1em>^(.7){\pi_1\x 1}
&
{\Tt\x\Tt}_{(4)}
\ar@{.)}@!+<-1em,1.8333em>;[0,0]+<-1.5em,.8333em>^(.0){(0,\bs)}
\ar@{.)}@!+<1.5em,1.8333em>;[0,0]+<1em,.83333em>^(.0){(1,\bs)}
\ar@{.)}@!+<2.8333em,0em>;[0,0]+<2em,0em>_(.0){(0,\terminal)}
\ar@{.)}@!+<2.8333em,-1.5em>;[0,0]+<1.6667em,-1em>_(.0){(1,\terminal)}
\avl{l}{.4}{\pii{23}}\\
{\N\x\Tt}_{(5)} \av{ru}{\pii{12}}
\ar@{.)}@!+<1.5em,-1.8333em>;[0,0]+<1em,-0.8333em>_(.0){(0,0)} &
{\Tt\x\Tt}_{(6)}
\ar@{.)}@!+<-1em,1.8333em>;[0,0]+<-1.5em,.8333em>^(.0){(0,0)}
\ar@{.)}@!+<1.5em,1.8333em>;[0,0]+<1em,.8333em>^(.0){(0,1)}
\ar@{.)}@!+<-1em,-1.8333em>;[0,0]+<-1.5em,-.8333em>_(.0){(1,0)}
\ar@{.)}@!+<1.5em,-1.8333em>;[0,0]+<1em,-.8333em>_(.0){(1,1)}
\au{r}{\pii{23}}&
{\N\x\Tt\x\Tt}_{(7)} \avx{r}{\alpha}&
{\N\x\N}_{(8)}
}
\]
Here, $\alpha=((\suc\circ\pi_1)\x\pi_2)\inv\circ\pi_{12}$.
Then, for a cross section $s$ of the diagram with $s(S_1)=A$,
\begin{align*}
s(S_5)=&\{(i+1,x)\|(i,x)\in A\}\cup\{(0,0)\},\\
s(S_2)=&\{(i+1,x,y)\|(i,x), (i+1,y)\in A\}\cup\{(0,0,y)\|(0,y)\in A\},\\
s(S_4)=&\{(x,y)\|x\in\2, y\in B\},\\
s(S_3)=&\{(i,x,y)\in s(S_2)\|x\in\2, y\in B\}\\
=&\{(i+1,x,y)\|(i,x), (i+1,y)\in A, x\in\2, y\in B\}\cup\{(0,0,y)\|(0,y)\in A, y\in B\},\\
s(S_6)=&\{(x,y)\|x,y\in\2\},\\
s(S_7)=&\{(i,x,y)\in s(S_2)\|x,y\in\2\}\\
=&\{(i+1,x,y)\|(i,x), (i+1,y)\in A, x,y\in\2\}\cup\{(0,0,y)\|(0,y)\in A, y\in\2\}\\
s(S_8)=&\{(i,x),(i+1,0),(i+1,1)\|(i,x),(i+1,y)\in A,x\in\2, y\in B\}\cup\\
&\{(0,0),(0,1)\|(0,y)\in A, y\in B\}\cup\{(i,x)\|(i,x),(i+1,y)\in A,x,y\in\2\},
\end{align*}
where $B=\{\bs,\terminal\}$.
Thus, $s(S_8)$ is the result of applying $\theta$ to $s(S_1)$.
\end{proof}

\begin{lem}\label{lemma:delta}
Let $\delta$ be a map from $\m\x\n$ to $\m\x\n\x\2$.
Then $\delta$ can be represented by a diagram generated by $\MN$.
\end{lem}
\begin{proof}
Let $k$ be the larger of $n$ and $m$.
The diagram has three parts: $S_{\mathrm{A},i}=\N$ for $i = 0,\cdots,k-1$;
$S_{\mathrm{B},i,j}=\m\x\n\x\m\x\n\x\2$ for $(i,j)\in\m\x\n$;
and $S_{\mathrm{C},0}=\m\x\n, S_{\mathrm{C},1}=\m\x\n\x\2$.

The first part is as follows:
\[
\diagram{
{\N}_{(\mathrm{A},0)} \aiul{0} \aul{r}{.45}{\suc} & {\N}_{(\mathrm{A},1)} \au{r}{\suc}  &
{\cdots} \aul{r}{.35}{\suc} & {\N}_{(\mathrm{A},k-1)}
}
\]
This part has no incoming arrows from other parts.
Thus, any cross section $s$ has $s(S_{\mathrm{A},i})=\{i\}$ for $i=0,1,\cdots,k-1$.

The other parts are organized thus:
\[
\diagram[@R=1em]{
{\m\x\n}_{(\mathrm{C},0)} \aul{r}{.35}{\pii{12}} \aul{rdd}{.45}{\pii{12}} &
 {\m\x\n\x\m\x\n\x\2}_{(\mathrm{B},0,0)}  \aux{rdd}{\pi_{345}}\\
 & {\vdots} \\
 & {\m\x\n\x\m\x\n\x\2}_{(\mathrm{B},m-1,n-1)} \auxl{r}{.6}{\pi_{345}} &
 {\m\x\n\x\2}_{(\mathrm{C},1)}
}
\]
with maps, for each $(i,j)\in\m\x\n$, as follows:
\[
\diagram[@C=4em]{
{\N}_{(\mathrm{A},i)} \aul{r}{.3}{\pii{1}} &
{\m\x\n\x\m\x\n\x\2}_{(\mathrm{B},i,j)} &
{\N}_{(\mathrm{A},\delta_1(i,j))} \avl{l}{.3}{\pii{3}} \\
{\N}_{(\mathrm{A},j)} \aul{ur}{.4}{\pii{2}} &
{\N}_{(\mathrm{A},\delta_2(i,j))} \av{u}{\pii{4}} &
{\N}_{(\mathrm{A},\delta_3(i,j))} \avl{ul}{.4}{\pii{5}}
}
\]
where $\delta(i,j)=(\delta_1(i,j),\delta_2(i,j),\delta_3(i,j))$.

Let $s$ be a cross section of this diagram.
Then for each $(i,j)\in\m\x\n$,
\begin{align*}
s(S_{\mathrm{B},i,j})&=
\pii{12}(s(S_{\mathrm{C},0}))\cap \pii{1}(s(S_{\mathrm{A},i}))\cap \pii{2}(s(S_{\mathrm{A},j}))\\
&\hspace{5em}\cap \pii{3}(s(S_{\mathrm{A},\delta_1(i,j)}))\cap
\pii{4}(s(S_{\mathrm{A},\delta_2(i,j)})) \cap \pii{5}(s(S_{\mathrm{A},\delta_3(i,j)})) \\
&=\{(i,j,\delta_1(i,j),\delta_2(i,j),\delta_3(i,j))\in\m\x\n\x\m\x\n\x\2\|(i,j)
\in s(S_{\mathrm{C},0})\}\\
&=
\begin{cases}
\{(i,j,\delta_1(i,j),\delta_2(i,j),\delta_3(i,j))\} & \text{if }(i,j)\in s(S_{\mathrm{C},0})\\
\emptyset & \text{if }(i,j)\notin s(S_{\mathrm{C},0})
\end{cases}
\end{align*}
and thus
\begin{align*}
s(S_{\mathrm{C},1})&=\bigcup_{(i,j)\in s(S_{\mathrm{C},0})}\pi_{345}(s(S_{\mathrm{B},i,j}))\\
&=\{\delta(i,j)\|(i,j)\in s(S_{\mathrm{C},0})\}=\delta(s(S_{\mathrm{C},0})).
\end{align*}
\end{proof}

\begin{lem}\label{psi}
The map $\psi:\2^{\N\x\Tt\x Q\x\2\x\N}\to\2^{\N\x Q\x\N}$ defined by
\begin{align*}
\psi(A)&=\{(0,q,k)\| (0,x,q,0,k)\in A\}\\
&\hspace{1em}\cup\{(i+1,q,k)\|(i,x,q,1,k)\in A\}\cup\{(i,q,k)\|(i+1,x,q,0,k)\in A\}
\end{align*}
can be represented by a diagram generated by $\MN$.
\end{lem}
\begin{proof}
Consider the diagram
\[
\diagram[@C=6.2em]{
& & {\N\x Q\x\N}_{(3)} \av{d}{(\suc\circ\pi_1)\x\pi_{23}}\\
{\N\x\Tt\x Q\x\2\x\N}_{(1)} \au{r}{\pi_{1345}} & {\N\x Q\x\2\x\N}_{(2)}
\au{ur}{(\pi_{12}\x 1\x\pi_3)\inv} \au{r}{((\suc\circ\pi_1)\x\pi_2\x 0\x\pi_3)\inv}
\au{dr}{\hspace{1em}(0\x\pi_1\x 0\x\pi_2)\inv}
& {\N\x Q\x\N}_{(4)}\\
& & {Q\x\N}_{(5)} \av{u}{0\x\pi_{12}}
}
\]
Then we have
\begin{align*}
s(S_2)&=\{(i,q,v,k)\|(i,x,q,v,k)\in s(S_1)\},\\
s(S_3)&=\{(i,q,k)\|(i,x,q,1,k)\in s(S_1)\},\\
s(S_5)&=\{(q,k)\|(0,x,q,0,k)\in s(S_1)\},\\
s(S_4)&=\{(i+1,q,k)\|(i,x,q,1,k)\in s(S_1)\}\cup\{(i,q,k)\|(i+1,x,q,0,k)\in s(S_1)\}\\
&\hspace{1em}\cup\{(0,q,k)\| (0,x,q,0,k)\in s(S_1)\}.
\end{align*}
Thus we have $s(S_4)=\psi(s(S_1))$.
\end{proof}

\begin{lem}
The map $\chi:\2^{\N\x\N}\to\2^{\N\x\N}$ such that
\[
\chi(\{(i,k)\})=\{(j,k)\|j\in\N, j\neq i\}
\]
can be represented by a diagram generated by $\MN$.
\end{lem}
\begin{proof}
Consider the diagram
\[
\diagram{
{\N\x\N}_{(1)} \aux{r}{s_1}\aux{d}{s_2} & {\N\x\N}_{(2)} \aux{d}{\id} \arlooprx{s_1} \\
{\N\x\N}_{(3)} \aux{r}{\id} \arlooplx{s_2} & {\N\x\N}_{(4)}
}
\]
Here, $s_1=(\suc\circ\pi_1)\x\pi_2$ and $s_2=((\suc\circ\pi_1)\x\pi_2)\inv$.
Then if $s(S_1)=\{(i,k)\}$,
\begin{align*}
s(S_2)&=\{(j,k)\|j=i+1,i+2,\cdots\},\\
s(S_3)&=\{(j,k)\|j=i-1,i-2,\cdots,0\},\\
s(S_4)&=\{(j,k)\|j\in\N, j\neq i\}.
\end{align*}
\end{proof}

\begin{lem}\label{deltastar}
The map $\delta^*:\2^{\N\x\Tt\x\Qt\x\N}\to\2^{\N\x\Tt\x\Qt\x\N}$
with the properties in the proof of the theorem can be represented by a diagram
generated by $\MN$.
\end{lem}
\begin{proof}
Consider the diagram:
\[
\diagram[@C=4.6em@R=3em]{
{\N\x\Tt\x\Qt\x\N}_{(1)} \aul{r}{.47}{\pi_{123}\x(\suc\o\pi_4)} \av{d}{\pi_{123}\x(\suc\o\pi_4)} &
{\N\x\Tt\x\Qt\x\N}_{(2)} \aul{r}{0.45}{\pi_1\x(\delta\o\pi_{23})\x\pi_4} \av{d}{\chi\o\pi_{14}}
\ar@{.>}@!+<1.5em,-1.8333em>;[0,0]+<1em,-.8333em>_(.0){\pii{3}\o(\Qt\setminus\{\qe,\qA\})} &
{\N\x\Tt\x\Qt\x\2\x\N}_{(3)} \avxl{dddl}{.3}{\psi} \aux{d}{\pi_{125}} \\
{\N\x\Tt\x\Qt\x\N}_{(4)} \av{d}{\pi_{134}} \aul{dr}{0.4}{\pi_{124}}
\ar@{.>}@!+<2em,1.8333em>;[0,0]+<1.5em,.8333em>^(.0){\pii{3}\o\{\qe,\qA\}} &
{\N\x\N}_{(5)} \au{d}{\pii{13}}&
{\N\x\Tt\x\N}_{(6)} \au{d}{\pii{124}} \\
{\N\x\Qt\x\N}_{(7)} \av{d}{\phi}
\ar@{.>}@!+<2em,1.8333em>;[0,0]+<1.5em,.8333em>^(.0){\pii{2}\o\qA} &
{\N\x\Tt\x\N}_{(8)} \auxl{ur}{.3}{\id} &
{\N\x\Tt\x\Qt\x\N}_{(9)} \au{d}{\eta}\\
{\N\x\Qt\x\N}_{(10)} \aux{r}{\id} &
{\N\x\Qt\x\N}_{(11)} \au{ur}{\pii{134}}
\ar@{.)}@!+<-1.5em,1.8333em>;[0,0]+<-1em,.8333em>_(.0){\pii{2}\o\qe} &
{\N\x\Tt\x\Qt\x\N}_{(12)}
}
\]
For $k\in\N$, let us assume that $s(S_1)$'s $k$'th row $R_k$ is proper.
We see that $s(S_2)$ contains the single element $(i,x,q,k), q\neq\qA$ in $R^2_k$ and $s(S_4)$ the rest,
with the step number $k$ incremented.
Note that $\{\qe,\qA\}$ and $\Qt\setminus\{\qe,\qA\}$ in the diagram denotes the constant union maps
$\1\to\Qt$ whose images are $\{\qe,\qA\}$ and $\Qt\setminus\{\qe,\qA\}$.

The map $\chi\circ\pi_{14}$ sends to
\[
s(S_5)=\bigcup_{(i,x,q,k)\in s(S_2)}\{(j,k)\|j\neq i\}.
\]
Thus, $s(S_5)$ contains exactly the position-step pair $(i,k)$ such that there is no element
$(i,x,q,k)$ in $s(S_2)$.
The element $(i,x,q,k)$ in $R^2_k$ indicates that the head is at position $i$ on step $k$.
Thus $s(S_5)$ contains the position-step pair everywhere except where the head is,
and $s(S_8)$ contains the symbols that are not under the head.

The symbol-state pair $(x,q)$ from $(i,x,q,k)$ in $R^2_k$ is fed to $\delta$ and the result is
\[
s(S_3)=\{(i,x',q',v,k)\| (i,x,q,k)\in s(S_2), \delta((x,q))=(x',q',v)\},
\]
which contains the symbol $x'$ to replace the one at the current position $i$, the new state $q'$,
and the direction $v$ to move the head.
The new symbol is sent to $s(S_6)$ where it is combined with the symbols that are not changed,
which are contained in $s(S_8)$.
The direction to move the head is indicated by $v$: if $v=0$, the head is to be moved
to the left, and if $v=1$ to the right.
The map $\psi$ in Lemma \ref{psi} is defined so that the position is moved accordingly, while
also taking care of the case $(i,v)=(0,0)$, when the head is indicated to move to the left and
the position is $0$, and therefore has to stay at $0$.
Also, without loss of generality, we can assume that $\delta$ satisfies $\delta((x,\qR))=(x,\qR,0)$,
which makes the simulated head to move to the left without changing any symbol
after the rejection state is reached, until it reaches the leftmost position, where it stays.

Coming back to $s(S_4)$,
\[
s(S_7)=\{(i,\qA,k)\| (i,x,\qA,k)\in s(S_4)\},
\]
takes position-step pairs with state $\qA$.
The map $\phi:\2^{\N\x\Qt\x\N}\to\2^{\N\x\Qt\x\N}$ is defined so that
\begin{align*}
s(S_{10})=\phi(s(S_7))&=\{(i,\qA,k)\| (i+1,\qA,k)\in s(S_7)\}\\
&\hspace{1em}\cup\{(i,\qA,k)\| (i,\qA,k)\in s(S_7)\}\\
&\hspace{1em}\cup\{(i+1,\qA,k)\| (i,\qA,k)\in s(S_7)\},
\end{align*}
which propagates the positions with accept state $\qA$.
This can be represented by
\[
\diagram[@C=10em]{
{\N\x\Qt\x\N} \aux{r}{\id}
\ar@{.)}@(ur,ul)!UR+<-.8333em,0em>;[r]!UL+<.8333em,0em>_{((\suc\o\pi_1)\x\pi_{23})\inv}
\ar@{.)}@(dr,dl)!DR+<-.8333em,0em>;[r]!DL+<.8333em,0em>^{(\suc\o\pi_1)\x\pi_{23}} &
{\N\x\Qt\x\N}
}
\]

Now, the position-state-step triples are combined in $s(S_{11})$:
\[
s(S_{11})=\psi(s(S_3)) \cup s(S_{10}) \cup \{(i,\qe,k)\|i,k\in\N\},
\]
where the last part ensures that the state $\qe$ is in $R^1_k$ everywhere there is a symbol.

The states and the symbols are combined in $s(S_9)$:
which is combined with the symbols to make
\[
s(S_9)=\{(i,x,q,k)\| (i,x,k)\in s(S_6), (i,q,k)\in s(S_{11})\}.
\]

Finally, the map $\eta:\2^{\N\x\Tt\x\Qt\x\N}\to\2^{\N\x\Tt\x\Qt\x\N}$ takes care of the case
when the head moves past the right end of the string:
\begin{align*}
\eta((i,\terminal,q,k))&=\{(i,\bs,q,k), (i+1,\terminal,\qe,k)\},\;\;(q\neq\qe, q\neq\qA),\\
\eta((i,x,q,k))&=\{(i,x,q,k)\}\;\;\text{otherwise.}
\end{align*}
This can be represented by:
\[
\diagram[@C=6em]{
{\N\x\Tt\x\Qt\x\N} \au{r}{\id} \au{rd}{\id} \au{rdd}{\id} &
 {\N\x\Tt\x\Qt\x\N}  \aux[@<.25em>]{r}{\pi_1\x\bs\x\pi_{34}}
\avx[@<-.25em>]{r}{(\suc\o\pi_1)\x\terminal\x\qe\x\pi_4}
\ar@{.>}@!+<-.3333em,2em>;[0,0]+<-.8333em,.8em>^(.0){\pii{23}\o(\terminal\x(\Qt\setminus\{\qe,\qA\}))}
& {\N\x\Tt\x\Qt\x\N} \\
& {\N\x\Tt\x\Qt\x\N} \avx[@<-.25em>]{ru}{\id}
\ar@{.>}@!+<-.3333em,1.8333em>;[0,0]+<-.8333em,.8333em>^(.0){\pii{2}\o(\Tt\setminus\{\terminal\})}\\
& {\N\x\Tt\x\Qt\x\N} \avx[@<-.25em>]{ruu}{\id}
\ar@{.>}@!+<.6667em,-1.8333em>;[0,0]+<.8333em,-.8333em>_(.0){\pii{3}\o\{\qe,\qA\}}
}
\]

\end{proof}

\begin{thm}\label{KCTheorem}
For any universal Turing machine $U$, there exists a constant $c_U\in \N$ such that
\[
I(\bsigma\,|\MN)\leq 6K_U(\sigma)+c_U
\]
for any binary string $\sigma\in\2^*$.
\end{thm}
\begin{proof}
By Theorem~\ref{TMEmulate}, there exists a diagram $(\S,\S',\M)$ generated by $\MN$
that emulates $U$.
By definition, there exists $p\in\2^*$ such that $U(p)=\sigma$ and $|p|=K_U(\sigma)$.
There are sets $S=\N\x\N$ and $T=\N\x\N$ in $\S$, as $S_1$ and $S_4$ in \eqref{TMEmuDiagram},
such that $s(T)=\bar{\sigma}$ if $s(S)=\bar{p}$, for any cross section $s$ of $(\S,\S',\M\|s_\1)$.
Add to the diagram the following:
\[
\diagram{
{\N}_{(\mathrm{A},0)} \aiul{0} \au{r}{\suc} & {\N}_{(\mathrm{A},1)} \au{r}{\suc}  &
{\cdots} \aul{r}{.35}{\suc} & {\N}_{(\mathrm{A},|p|)}
}
\]
with maps from $S_{\mathrm{A},i}$ to $S$:
\[
\diagram[@C=4em]{
{\N}_{(\mathrm{A},i)} \aux{r}{\id\x p[i]} & {\N\x\N}
}
\]
for $i=0,1,\dots,|p|-1$, where $\id\x p[i]:\N\to\N\x\N$ denotes the map
$\id\x 0$ if $p[i]=0$ and $\id\x (\suc\circ 0)$ if $p[i]=1$.
Finally, we add maps
\[
\diagram[@C=4em]{
{\N}_{(\mathrm{A},|p|)} \aux[@<.25em>]{r}{\id\x 0} \avx[@<-.25em>]{r}{\id\x 1} & {\N\x\N}
}
\]
from $S_{\mathrm{A},|p|}$ to $S$.
Then any cross section $s$ of the resulting diagram satisfies $s(S)=\bar{p}$, and thus $s(T)=\bar{\sigma}$.

Maps $\suc$ and $\id\x p[i]$ is added for each symbol in $p$.
The size of $\suc$ is 1 and that of $\id\x (\suc\circ 0)$ is 5.
Thus the size of the added diagram is at most $6|p|+\mathrm{const}$.
Since the only part of the diagram that depends on $\sigma$ is the part added here,
the existence of this diagram proves the theorem.
\end{proof}

\begin{cor}
For any binary string $\sigma\in\2^*$, $I(\bsigma\,|\MN)$ is bounded from above by
\[
c\cdot \min_{\substack{M:\mathrm{TM}, p\in \2^*\\M(p)=\sigma}}(|p|+|\Theta_M||Q_M|),
\]
where $c$ is a constant and $\Theta_M$ and $Q_M$ are the sets of internal symbols and states of
Turing machine $M$, respectively.
\end{cor}
\begin{proof}
Use the same construction as in the proof of Theorem~\ref{KCTheorem} with any (possibly non-universal) Turing machine.
It can be seen from the proof of the theorem and Lemma~\ref{lemma:delta} that the size of the diagram emulating a
Turing machine $M$ is $O(|\Theta_M||Q_M|)$.
\end{proof}

\subsection{Reading Cross Section}

Conversely, if a finite diagram generated by $\MN$ can represent $\bsigma$ for some
binary string $\sigma$, there is a Turing machine that produces $\sigma$ and terminate.

Let $X$ be a set of Boolean variables, i.e., variables that take values in $\2$.
We use the standard notion like logical AND ($\wedge$) and OR ($\vee$),
treating $1$ as true and $0$ as false.
By an assignment to $X$, we mean a map $f:X\to\2$, which assigns $0$ or $1$ to each
variable in $X$.
If $Y\subset X$, there is a natural map $\2^X\to\2^Y$ by restriction that maps an
assignment to $X$ to an assignment to $Y$.
A logical constraint $\chi$ on a set $X$ of Boolean variables is a map $\chi:\2^X\to\2$,
which is said to be satisfied by an assignment $f\in\2^X$ to $X$ if $\chi(f)=1$.
For a logical constraint $\chi$ on $X$, an assignment $g$ to a subset $Y$ of $X$
is said to satisfy $\chi$ if every assignment that restricts to $g$ satisfies it.
For a set $C$ of constraints on $X$, we say an assignment to $X$ satisfies $C$ if it
satisfies all the constraints in $C$.

\begin{thm}\label{DSimulate}
Let $(\S, \S',\M)$ be a finite diagram generated by $\MN$
such that the data $(\S$, $\S'$, $\M$, $\{\1\}$, $s_{\1}$, $S)$ represents $\bsigma\subset S=\N\x\N$ for
a binary string $\sigma$,
where $s_{\1}$ is the cross section of $\{\1\}$ defined by $s_{\1}(\1)=\1$.
Then there exists a Turing machine that takes an encoded description of any such diagram and produces $\sigma$.
Therefore, for any universal Turing machine $U$, there exist constants $d_U, e_U\in \N$ such that
\[
K_U(\sigma)\leq d_U I(\bsigma\,|\MN)+ e_U
\]
for any binary string $\sigma\in\2^*$.

\end{thm}
\begin{proof}
For each element $t$ of each set $T$ in $\S$, we define a Boolean variable $x_t$.
Also, for each set $T$ in $\S$, each map $\varphi\in\In(T)$, and each element $t$ of $T$,
we define a Boolean variable $y^\varphi_t$.
Let us denote the set of the variables by $X$.
The variable $x_t$ is for indicating if $t$ is in the cross section, and $y^\varphi_t$ is to
indicate if $t$ is in the image of the cross section by $\varphi$.
We define a set $C$ of constraints on $X$ to establish a one-to-one correspondence
between the cross sections in $\Gm(\S,\S',\M\|s_{\1})$ and the assignments to $X$ satisfying $C$.
We let $C$ contain the following constraints:
\begin{enumerate}[i)]

\item \label{mapconst} For each map $\varphi\in\M$:
\begin{enumerate}[a)]

\item If $\varphi$ is the power map $f:\2^U\to\2^T$ of $f:U\to T$, $C$ contains the
constraint $y^\varphi_t=0$ for each $t\in T\setminus f(U)$ and, for each $t\in f(U)$,
\[
y^\varphi_t=\bigvee_{u\in f\inv(t)}x_u.
\]
\item If $\varphi$ is the inverse map $f\inv:\2^U\to\2^T$ of $f:T\to U$, $C$ contains the constraint
\[
y^\varphi_t=x_{f(t)}
\]
for each $t\in T$.
\end{enumerate}

\item For each set $T$ in $\S\setminus\S'$ such that $\In(T)\neq\emptyset$ and each $t$ in $T$,
$C$ contains the constraint
\[
x_t=\bigwedge_{\varphi\in \In(T)}y^\varphi_t.
\]

\item For each set $T$ in $\S'$ such that $\In(T)\neq\emptyset$ and each $t$ in $T$,
$C$ contains the constraint
\[
x_t=\bigvee_{\varphi\in \In(T)}y^\varphi_t.
\]

\item For the variable $x_0$ corresponding to the element $0\in\1$ that appears in the data
$(\S,\S',\M, \{\1\}, s_{\1}, S)$, $C$ contains the constraint $x_0=1$.
\end{enumerate}
Note that for each variable $x\in X$, there is exactly one constraint in $C$ with that variable on the LHS.
We call it $\chi_x$.
Note also that any conjunction in $C$ has only a finite number of variables and that
there are only countably many variables in $X$.
We fix a one-to-one correspondence $\nu:X\to\N$ for later use.

For a cross section $s$ of $\S$, we define an assignment $g_s$ to $X$ by defining $g_s(x_t)=1$ if and only if $t\in s(T)$ for each
$T\in\S, t\in T$
and $g_s(y^\varphi_t)=1$ if and only if $t\in\varphi(s(\dm(\varphi)))$.
Then $g_s$ satisfies $C$ if and only if $s$ is a cross section in $\Gm(\S,\S',\M\|s_{\1})$,
as follows.

First, suppose that $g_s$ satisfies $C$.
If $T\in\S\setminus\S'$, then $g_s(x_t)=\bigwedge_{\varphi\in \In(T)}g_s(y^\varphi_t)$ and
thus $t\in s(T)$ (i.e., $g_s(x_t)=1$) if and only if
$t\in\bigcap_{\varphi\in\In(T)}\varphi(s(\dm(\varphi)))$.
If $T\in\S'$, then $g_s(x_t)=\bigvee_{\varphi\in \In(T)}g_s(y^\varphi_t)$;
thus $t\in s(T)$ (i.e., $g_s(x_t)=1$) if and only if
$t\in\bigcup_{\varphi\in\In(T)}\varphi(s(\dm(\varphi)))$.
This proves that $s$ is in $\Gm(\S,\S',\M)$.
Finally, the constraint $g_s(x_0)=1$ for the variable $x_0$ corresponding to the element $0\in\1$
ensures that $s(\1)=\1$, and thus $s$ is in $\Gm(\S,\S',\M\|s_{\1})$.
Conversely, suppose that $s$ is in $\Gm(\S,\S',\M\|s_{\1})$.
If $T\in\S\setminus\S'$, $g_s(x_t)=1$ (i.e., $t\in s(T)$) if and only if $t\in\bigcap_{\varphi\in\In(T)}\varphi(s(\dm(\varphi)))$;
since $g_s(y^\varphi_t)=1$ if and only if $t\in\varphi(s(\dm(\varphi)))$, $g_s(x_t)=\bigwedge_{\varphi\in \In(T)}g_s(y^\varphi_t)$.
If $T\in\S'$, $g_s(x_t)=1$ (i.e., $t\in s(T)$) if and only if $t\in\bigcup_{\varphi\in\In(T)}\varphi(s(\dm(\varphi)))$;
thus $g_s(x_t)=\bigvee_{\varphi\in \In(T)}g_s(y^\varphi_t)$.
Finally, $g_s(x_0)=1$ is satisfied since $s(\1)=\1$.

We define a subset $X_i$ of $X$ for $i\in\N$ by
\begin{align*}
X_0&=\{x_0\},\\
X_{i+1}&=X_i\cup\{x\in X\|x\text{ is forced by }X_i\}.
\end{align*}
Here, a variable $x\in X$ is said to be forced by a subset $Y$ of $X$ either if $\chi_x$ is
a disjunction and at least one variable on its RHS is in $Y$ or if $\chi_x$ is a conjunction and all
the variables on its RHS is in $Y$.
Note that, if an assignment $g$ to $X$ satisfies $C$, $g(x_0)=1$ and,
by following the definition, any variable $x$ in $X_i$ has $g(x)=1$.

We define a function $\tau$ on $X$:
\[
\tau(x)=
\begin{cases}
i & \text{if }x\in X_i\setminus X_{i-1}\\
\infty &\text{otherwise}
\end{cases}
\]
and using $\tau$, we define an assignment $h$ to $X$ by:
\[
h(x)=1\;\; \Longleftrightarrow \;\;\tau(x)<\infty.
\]
Then $h$ satisfies $C$.
To see this, assume that there is a constraint $\chi\in C$ that is not satisfied by $h$.
If $\chi$ is of the form $x=\wedge y_j$ with a finite number of $y_j$'s,
either $h(x)=1$ and $h(y_j)=0$ for some $y_j$, or $h(x)=0$ and $h(y_j)=1$ for all $y_j$.
Neither is possible by the definition of $h$:
if $h(y_j)=0$ for some $y_j$, $\tau(y_j)=\infty$ and thus $\tau(x)=\infty$;
if $\tau(y_j)<\infty$ for all $y_j$, then all $y_j$'s are in $X_k$, where $k$ is the maximum
of $\tau(y_j)$, and thus $\tau(x)=k+1$.
If $\chi$ is of the form $x=\vee y_j$,
either $h(x)=1$ and $h(y_j)=0$ for all $y_j$, or $h(x)=0$ and $h(y_j)=1$ for some $y_j$.
These are not possible either:
if $h(y_j)=0$ and thus $\tau(y_j)=\infty$ for all $y_j$, then $\tau(x)=\infty$;
if $h(y_j)=1$ for some $y_j$, $\tau(y_j)<\infty$ and $\tau(x)\leq\tau(y_j)+1$.

For any $x$ in $X$ with $\tau(x)<\infty$, define a finite subset $Y_x$ of $X_{\tau(x)}$
as follows:
\begin{align*}
x&\in Y_x,\\
y\in Y_x,\; \chi_y=\text{``}y=\wedge y_j\text{'' }&\Rightarrow \forall y_j\in Y_x,\\
y\in Y_x,\; \chi_y=\text{``}y=\vee y_j\text{'' }&\Rightarrow \arg\min_{y_j, \tau(y_j)<\tau(y)} \nu(y_j)\in Y_x.
\end{align*}
Since only a finite number of variables appear in any conjunction in $C$,
for a variable $y$ already in $Y_x$, each rule adds at most a finite number of variables $y_j$,
which all satisfies $\tau(y_j)<\tau(y)$.
Thus there are only a finite number of variables in $Y_x$.
Also, if we denote $Y_x^i=Y_x\cap X_i$, each variable in $Y_x^i$ is forced by $Y_x^{i-1}$
 for $i=1,\cdots,\tau(x)$.

\newcommand{\numright}[1]{\makebox[1.1667em][r]{#1}}

Consider the following algorithm:

\begin{tabbing}
\tab{\sc IsOne}$(x)$\\
\tab\numright{1}\tab\={\bf for} \= {\bf each }increasingly large finite subset $Z\subset X$ such that $x_0,x\in Z$\\
\tab\numright{2}\>\>$Y_0\leftarrow\{x_0\}$ \\
\tab\numright{3}\>\>{\bf for }\= $i=1,2,\cdots$ \\
\tab\numright{4}\>\>\>$Y_i\leftarrow Y_{i-1}$ \\
\tab\numright{5}\>\>\>{\bf for }\={\bf each} $z\in Z$\\
\tab\numright{6}\>\>\>\>add $z$ to $Y_i$ if $z$ is forced by $Y_{i-1}$\\
\tab\numright{7}\>\>{\bf until} $Y_i=Y_{i-1}$\\
\tab\numright{8}\>{\bf until} $x\in Y_i$ for some $i$
\end{tabbing}
On line 1, we make sure the increasing subset $Z$ will contain each variable eventually by
adding variable $\nu\inv(k)$ for $k=0,1,2,\cdots$.

If $\tau(x)<\infty$, {\sc IsOne}$(x)$ terminates in finite steps, since
$Z$ will eventually include $Y_x$.
Conversely, if {\sc IsOne}$(x)$ terminates in finite steps, the set $Y_i$ that contains $x$
is a subset of $X_i$.
Thus $\tau(x)<\infty$.

Let us assume $s\in\Gm(\S,\S',\M\|s_{\1})$.
Since the data $(\S$, $\S'$, $\M$, $\{\1\}$, $s_{\1}$, $S)$ represents $\bsigma\subset S=\N\x\N$,
any cross section $s'$ in $\Gm(\S,\S',\M\|s_{\1})$ satisfies $s'(S)=\bsigma$.
Thus, for any assignment $g$ to $X$ that satisfies $C$, $g(t)=g_s(t)$ for any $t\in S$.
In particular, we have $h(t)=g_s(t)$ for any $t\in S$.
This means that $\tau(x_{(i,b)})<\infty$ for $(i,b)\in S$ if and only if $\sigma[i]=b$.

Now, we define a Turing machine $M$ that executes the following algorithm:
\begin{tabbing}
\tab\numright{1}\tab\=Given $(\S,\S',\M, \{\1\}, s_{\1}, S)$\\
\tab\numright{2}\tab Let $\rho$ be an empty string variable\\
\tab\numright{3}\tab\={\bf for }\=$i=0,1,\cdots$\\
\tab\numright{4}\>\>{\bf run }\= {\sc IsOne}$(x_t)$ in parallel for $t=(i,0), (i,1)$,
and $(i-1, 1-\rho[i-1])$ if $i>0$.\\
\tab\numright{5}\>\>\>{\bf until }one of the parallel processes terminates\\
\tab\numright{6}\>\>{\bf if } $i>0$ and {\sc IsOne}$(x_{(i-1, 1-\rho[i-1])})$ terminated\\
\tab\numright{7}\>\>\>{\bf then} pop $\rho[i-1]$ and terminate returning $\rho$\\
\tab\numright{8}\>\>{\bf else if} {\sc IsOne}$(x_{(i,0)})$ terminated\\
\tab\numright{9}\>\>\>{\bf then} $\rho[i]\leftarrow 0$\\
\tab\numright{10}\>\>{\bf else if} {\sc IsOne}$(x_{(i,1)})$ terminated\\
\tab\numright{11}\>\>\>{\bf then} $\rho[i]\leftarrow 1$
\end{tabbing}
For $i=0,1,\cdots,|\sigma|-1$, {\sc IsOne}$(x_{(i,\sigma[i])})$ terminates in finite steps but
the other processes do not.
For $i=|\sigma|$, either {\sc IsOne}$(x_{(i,0)})$ or {\sc IsOne}$(x_{(i,1)})$ can
terminate first.
For $i=|\sigma|+1$, only {\sc IsOne}$(x_{(i-1, 1-\sigma[i-1])})$ terminates.
Thus, $M$ always terminates returning $\sigma$.

Finally, let $(\S,\S',\M, \{\1\}, s_{\1}, S)$ be the diagram that has size $I(\bsigma\,|\MN)$.
Since the data can be encoded in a string of length $O(I(\bsigma\,|\MN))$ and
emulating $M$ on $U$ to run on the string produces $\sigma$,
there are constants $d_U, e_U\in \N$ such that
\[
K_U(\sigma)\leq d_U I(\bsigma\,|\MN)+ e_U
\]
that do not depend on $\sigma$.
\end{proof}

By Theorem~\ref{KCTheorem} and Theorem~\ref{DSimulate}, we can say that $I(\bsigma\,|\MN)$ is equivalent to
the Kolmogorov complexity $K_U(\sigma)$.
Thus, the formulation gives the larger class of objects a meaningful measure of information
that generalizes Kolmogorov complexity.

\section{Discussion and Conclusion}\label{sec:conclusion}

How is the information measure we defined different from Kolmogorov complexity in the general case, not the case in the previous section?
Let us take the class $\G$ of subsets of Euclidean plane $X$ as an example
and try to follow the common prescription: that is, we fix an encoding of the objects into strings; and then
we define the Kolmogorov complexity of the string encoding an object as its complexity.
Consider
\begin{equation}
S\stackrel{U}{\longrightarrow}\;S\stackrel{f'}{\longrightarrow} \G
\label{eq:old}
\end{equation}
where $S$ is the set of all strings and $U$ is the partial map defined by a universal Turing machine.
The map $f'$ is an encoding of objects by strings.
The common notion is that we should define the length of the shortest string $s$ such that $x=f'(U(s))$ as the Kolmogorov complexity of $x$.
However, as we mentioned in \ref{kolmocovers}, because the cardinality of $\G$ is larger than that of $S$,
we need to either encode only some of the objects, encode multiple objects by each string, or both.
To allow for both possibilities, we consider
\begin{equation}
S\stackrel{U}{\longrightarrow}\;S'\stackrel{f}{\longleftarrow} \G
\label{eq:old2}
\end{equation}
instead.
Here, $S'=S+\{\uparrow\}$, where $f(x)=\,\uparrow$ means ``$x$ is not encoded.''
For an object $x$, the length of the shortest string $s$ such that $f(x)=U(s)$ is defined as the amount $K_{f,U}(x)$ of
information in $x$, unless $f(x)=\,\uparrow$, in which case $K_{f,U}(x)=\infty$.

The question is: does there exist a map $f$ that makes $K_{f,U}(x)$ equivalent to $I(x\,|\M_\mathrm{E})$?
The trivial answer is: yes, we can define $f$ using $I(x\,|\M_\mathrm{E})$.
If $I(x\,|\M_\mathrm{E})$ is finite, we define $f(x)$ to be the first string $s$ (in some standard order)
such that the shortest string $s'$ with $U(s')=s$ has length $I(x\,|\M_\mathrm{E})$; otherwise we define $f(x)=\,\uparrow$.
But this only highlights the meaninglessness of this kind of discussion without
restricting $f$ in some way: as we noted in the introduction, just about any ``complexity'' can be realized this way.

Let us instead try to define a reasonable encoding $f$ to see if it gives rise to anything close to $I(x\,|\M_\mathrm{E})$.
First, a standard way to encode points in $X$ would be to identify $X$ with $\R^2$ and
encode the two coordinates of a point by an infinite sequence.
We can also encode a countable set of points by a sequence similarly by dovetailing between more and more points,
enumerating the digits for each to higher and higher precision.
To accommodate this, we again modify \eqref{eq:old2} a little and consider
\begin{equation}
S\stackrel{\tilde{U}}{\longrightarrow}\;\tilde{S}\stackrel{\tilde{f}}{\longleftarrow} \G
\label{eq:old3}
\end{equation}
Here, $\tilde{S}$ is the set of all infinite sequences.
If a sequence $s$ is computable, there is a Turing machine that, given a natural number $n$, prints the first $n$ symbols of $s$ and halts.
Let $\tilde{U}$ be the partial map that maps an encoding of such a Turing machine to the sequence.
We define the length of the shortest string $s'$ such that $\tilde{U}(s')=\tilde{f}(x)$ as $K_{\tilde{f},\tilde{U}}(x)$.
If there is no such string $s'$, we define $K_{\tilde{f},\tilde{U}}(x)=\infty$.
Note that this particular scheme is already quite different from ours.
According to this scheme, the information in a single point varies in a way irrelevant to anything
we may be interested in while thinking about subsets of a Euclidean plane.
If the sequence corresponding to the point is uncomputable, it has infinite information.
In contrast, for any point $p$ in $X$, $I(\{p\}\,|\M_\mathrm{E})=1$, as we saw in
section~\ref{sec:measure}.
We begin to see how we introduce unnecessary complications by trying to first encode everything by a string.

How about other, uncountable sets?
Most of geometric objects are uncountable subsets of $X$.
We can imagine encoding lines with a pair of points and circles with a point and a real number, etc., encoded as above.
That is, we {\it define} $\tilde{f}$ so that a line on $X$ going through a pair of points is mapped to a sequence
encoding two points and some flag indicating that it is a line.
This is a definition of $\tilde{f}$ by using characteristics, or regularity, of the subset.
However, the problem is that any interpretation of the sequence in terms of regularity must be incorporated in the
definition of $\tilde{f}$.
Thus, as we wish to add more patterns---lines, circles, repeated patterns---the definition of $\tilde{f}$
has to become more and more elaborate;
and we would wish to somehow define $\tilde{f}$ computationally.
However, this cannot be done because the spirit of using Kolmogorov complexity is that any information,
including any regularities and redundancies, is conveyed through a sequence and any computation
must be done by a Turing machine, strictly on the left side of $\tilde{S}$ in \eqref{eq:old3}.

In order to define $\tilde{f}$ computationally, we would have to define the notion of computation involving $\G$.
And, of course, that is exactly what our scheme does, where the situation is like
\begin{equation}
S\stackrel{g}{\longrightarrow}\D\stackrel{h}{\longrightarrow} \G
\label{eq:new}
\end{equation}
A given set of maps determines the possible set $\D$ of diagrams generated by the maps.
The map $h$ is the interpretation of diagrams into the ground representation.
The difference between the two schemes is on which side of the center the computation, or the ``decompression,''
takes place.
It occurs in $U$ in \eqref{eq:old3} and in $h$ in \eqref{eq:new}.

Thus, it is not easy to give a reasonable and simple definition of $\tilde{f}$ that makes
$K_{\tilde{f},\tilde{U}}(x)$ equivalent to $I(x|\M_\mathrm{E})$.
We cannot seem to define a simple encoding $\tilde{f}$ that reflects the informal notion of information in
the domain of subsets of Euclidean plane $X$.
Even with this single, relatively simple class of objects, we have these problems.
Remember that we cannot just say that there is an encoding and only talk about strings, as we pointed out repeatedly; if we insist encoding objects into strings, we need to define a concrete encoding for each class of objects.

The new representation provides a meta-definition for that purpose.
It allows specifying how to embed multiple computations in larger spaces in a versatile way.
As shown in the previous section, a Turing machine can be simulated by a diagram generated by $\{0,\suc\}$.
This can be isomorphically embedded in a diagram generated by $\M_\mathrm{E}$ by using $p$ and $\func{move}_v$ instead of
$0$ and $\suc$, where $p\in X$ and $\func{move}_v:X\to X$ is a move by a vector $v$, which can be
represented by a diagram.
Although a computation can be in other more direct forms, as in the examples in \ref{compexample}, this shows that universal 
computation can be embedded as a part of the representation.

Also, it allows separating the information that we wish to ignore from that which we measure.
As we have seen, this is important as some information we want to ignore is infinite.
The structure maps delineate the {\it a priori} parts of the objects whose information are to be ignored,
leaving only the information in the structure: hence the name ``structural information.''

Finally, the representation can provide a useful abstraction in designing real-world applications.
In dealing with high-dimensional signal-level objects, such as sounds, images, and sensory readings in robots, 
we need to model the process of finding patterns in raw signal and grouping them together to be described at a symbol level.
The information measure may be used to regularize the process of finding useful patterns in the ground level.
We can imagine implementing a system using this representation, which corresponds to defining an encoding $g$ of diagrams
shown in \eqref{eq:new}.

Given a fixed set of structure maps, it can represent infinite sets by finite symbolic expressions and sparse parameters.
That is, after fixing the sets and maps as primitives, diagrams can be symbolically represented just as graphs are.
While the fixed sets such as Euclidean spaces and the space of real numbers have to be implemented by
some system-dependant approximation such as floating point numbers,  the structure itself is preserved
irrespective of such approximations.
In the case of the line example, although points and vectors might be translated between systems varied in, say, resolution,
its structure, its ``lineness,'' so to speak, would survive.
In the current practice, an encoding of a line by such sparse data would have to be written into the code; leading
to the lack of generality and flexibility.

Also, diagrams can be naturally understood as modules to construct larger and more complex ones.
The set union and intersection are the most basic ways to combine them;
or we can put the output of one diagram into the input of another.
Ultimately, any computation can be used to combine them.
Our formulation provides a useful abstraction that allows us to model such structures in a uniform way
so that they can be manipulated automatically.

\vspace{2ex}
\noindent{\bf Conclusion}

\noindent In this paper, we introduced a uniform representation of general objects.
In an abstraction of the dense representation, which includes strings, bitmaps, and other raw data,
the objects are assumed to be given {\it a priori} as subsets of some sets.
The proposed representation uses a new construct called {\it diagram}
to represent objects with regularity through sparse parameters,
using the maps that characterize the space in which the objects are included as subsets.
Since the representation can emulate any computation, it can exploit any computable regularities in objects to compactly
describe them.
It is also general enough to represent random objects as raw data, making it possible
to interpolate between the dense and sparse representation.
There is a prescribed way to connect the description to the raw data;
thus, the representation is {\it grounded}.
In other words, the relationship between the parameters and the data is part of the representation
so that we can give raw data and then ask what sparser,
more structured representation is possible.
With the representation, we also defined a measure of information in the objects.
We proved that the measure is equivalent to Kolmogorov complexity in the case of strings.

To answer our question in the introduction, we would say that a subset $A$ of a set $X$ with
structures characterized by a set $\M$ of maps is a pattern if $I(A\,|\M)$ is significantly smaller than $|A|$
in the case $A$ is finite and, in the case $A$ is infinite, if $I(A\,|\M)$ is finite.
Even when $I(A\,|\M)$ is infinite, there can be a pattern in a statistical sense,
which we leave for future work.

\end{document}